\documentclass[12pt]{article}
\usepackage{amssymb,latexsym}
\usepackage[intlimits]{amsmath}
\usepackage{hyperref}
\usepackage{amsthm}
\usepackage{accents}
\usepackage{cite}
\usepackage{graphicx}
\usepackage{color}
\usepackage{cleveref}

\usepackage[normalem]{ulem}
\numberwithin{equation}{section}

\DeclareFontFamily{OT1}{pzc}{}
\DeclareFontShape{OT1}{pzc}{m}{it}{<-> s * [1.10] pzcmi7t}{}
\DeclareMathAlphabet{\mathpzc}{OT1}{pzc}{m}{it}
\usepackage[shortlabels]{enumitem}

\newtheorem{theorem}{Theorem}[section]
\newtheorem{lemma}[theorem]{Lemma}
\newtheorem{deflemma}[theorem]{Definition and Lemma}

\newtheorem{prop}[theorem]{Proposition}
\newtheorem{rmk}[theorem]{Remark}
\newtheorem{coroll}[theorem]{Corollary}

\newcommand\R{{\mathbb R}}

\renewcommand\d{\partial}

\newcommand{\definedas}{\mathrel{\raise.095ex\hbox{\rm :}\mkern-5.2mu=}}
\newcommand{\asdefined}{\mathrel{=\mkern-5.2mu}\raise.095ex\hbox{\rm :}\;}

\newcommand\beq{\begin{equation}}
\newcommand\eeq{\end{equation}}
\newcommand\ben{\begin{enumerate}}
\newcommand\een{\end{enumerate}}
\newcommand\bit{\begin{itemize}}
\newcommand\eit{\end{itemize}}

\newcommand{\Ric}{\operatorname{Ric}}
\newcommand{\Scal}{\operatorname{R}}
\newcommand{\sphere}{\mathbb{S}^{2}}
\newcommand\Crit{\operatorname{Crit}(u)}
\newcommand\CritN{\operatorname{Crit}(N)}
\newcommand\diver{\operatorname{div}}

\newcommand{\abs}[1]{\left\vert #1 \right\vert}

\newcommand{\sign}{\operatorname{sign}}

\title{Uniqueness of equipotential photon surfaces in $4$-dimensional static vacuum asymptotically flat spacetimes for positive, negative, and zero mass --\\ and a new partial proof of the Willmore inequality}
\author{Carla Cederbaum\thanks{cederbaum@math.uni-tuebingen.de, Mathematics Department}\,,  Albachiara Cogo\thanks{albachiara.cogo@uni-tuebingen.de, Mathematics Department}\,, and  Axel Fehrenbach\thanks{axel.fehrenbach@uni-tuebingen.de, Computer Science Department}\\Eberhard Karls Universit\"at T\"ubingen}
\date{}

\begin{document}
\maketitle
\vspace{-3ex}
\begin{abstract}
We present different proofs of the uniqueness of $4$-dimensional static vacuum asymptotically flat spacetimes containing a connected equipotential photon surface or in particular a connected photon sphere. We do not assume that the equipotential photon surface is outward directed or non-degenerate and hence cover not only the positive but also the negative and the zero mass case which has not yet been treated in the literature. Our results partially reproduce and extend beyond results by Cederbaum and by Cederbaum and Galloway. 

In the positive and negative mass cases, we give three proofs which are based on the approaches to proving black hole uniqueness by Israel, Robinson, and Agostiniani--Mazzieri, respectively. 

In the zero mass case, we give four proofs. One is based on the positive mass theorem, the second one is inspired by Israel's approach and in particular leads to a new proof of the Willmore inequality in $(\R^{3},\delta)$, under a technical assumption. The remaining two proofs are inspired by proofs of the Willmore inequality by Cederbaum and Miehe and by Agostiniani and Mazzieri, respectively. In particular, this suggests to view the Willmore inequality and its rigidity case as a zero mass version of equipotential photon surface uniqueness.
\end{abstract}

\section{Introduction}
Trapping of light is a very intriguing phenomenon in General Relativity, central to geometric optics and black hole imaging, and relevant in studying the dynamical stability of e.g.\ the Kerr spacetime. In some cases, trapped light rays can combine to so-called photon spheres, e.g.\ in the Schwarzschild spacetime. 

Here, a \emph{photon sphere} is a timelike hypersurface made out of trapped light, i.e., ruled by null geodesics bounded away from any black holes and naked singularities as well as from the asymptotic end of the spacetime under consideration. The prime example is the hypersurface $\{r=3m\}$ in the Schwarzschild spacetime of positive mass $m>0$. It turns out that this is indeed the only example of a photon sphere (see \Cref{sec:prelims} for a precise definition) arising as the inner boundary of a $4$-dimensional static vacuum asymptotically flat spacetime. In other words ``static vacuum photon spheres have no hair'', just like static vacuum black holes.

This was first proven by Cederbaum~\cite{Ced}, adapting Israel's approach to proving static vacuum black hole uniqueness~\cite{israel1967event}, under the technical assumption that the static lapse function regularly foliates the static spacetime. This approach is restricted to a single or in other words connected photon sphere. Subsequently, adapting Bunting and Masood-ul-Alam's method to proving static vacuum black hole uniqueness~\cite{BMuA}, Cederbaum and Galloway~\cite[Theorem 3.1]{carlagregpmt} proved the same result without the connectedness assumption on the photon sphere. However, they assume that the photon spheres are \emph{non-degenerate}, i.e., that the differential of the static lapse function does not vanish along the photon sphere. One of the goals of this paper is to show that the approaches to proving static vacuum black hole uniqueness by Robinson~\cite{Rob} and by Agostiniani and Mazzieri~\cite{Mazz} can also be adapted to proving uniqueness of static vacuum connected photon spheres, be they degenerate or non-degenerate.

Going beyond photon spheres to more general \emph{photon surfaces}, i.e., to timelike hypersurfaces ruled by null geodesics, Cederbaum--Jahns--Vi\v{c}\'anek-Mart\'inez~\cite[Theorems 3.7, 3.9, and 3.10]{CJV} proved that the Schwarzschild spacetime of any mass $m\in\R$ hosts a zoo of spherically symmetric photon surfaces, while Cederbaum and Galloway~\cite[Theorem 3.8 and Corollary 3.9]{CedGalSurface} proved that all photon surfaces in the Schwarzschild spacetime of any mass $m\in\R$ must be either spherically symmetric or pieces of certain timelike isotropic coordinate hyperplanes. Moreover, Cederbaum and Galloway~\cite[Theorem 4.1]{CedGalSurface} proved uniqueness of so-called ``outward directed equipotential'' photon surfaces, that is, any static vacuum asymptotically flat spacetime with an inner boundary consisting of outward directed equipotential photon surfaces must be isometric to a suitable piece of a Schwarzschild spacetime of positive mass. Their approach is similar in spirit to that in \cite[Theorem 3.1]{carlagregpmt}. Here, a photon surface is called \emph{equipotential} if the static lapse function is constant on each ``time-slice'' of the photon surface and \emph{outward directed} if the normal derivative of the static lapse function is strictly positive along the photon surface with respect to the normal pointing to infinity (see \Cref{sec:prelims} for precise definitions). Note that non-degenerate photon spheres are outward directed equipotential photon surfaces by \cite[Lemma 2.6, Remark 2.7]{carlagregpmt}. Also, being outward directed implies being non-degenerate.

The first main goal of this paper is to demonstrate uniqueness of $4$-dimensional static vacuum asymptotically flat spacetimes with a non-degenerate connected equipotential photon surface inner boundary. In particular, we will show that the outward directedness assumption made in \cite[Theorem 4.1]{CedGalSurface} is unnecessary in the connected case; this will allow us to prove equipotential photon surface uniqueness also in the ``negative mass case'', showing rigidity of negative mass Schwarzschild spacetimes. This part of our results has recently also been established by Cederbaum, Cogo, Leandro, and Paulo dos Santos~\cite{CCLP}.

The second main goal of this paper is to demonstrate that the non-degeneracy assumption made in \cite{carlagregpmt,CCLP} is unnecessary in the connected case in $4$ spacetime dimensions; this will allow us to prove equipotential photon surface uniqueness also in the so far unstudied ``zero mass case'', showing rigidity of the Minkowski spacetime. 

Altogether, we will prove the following partially new results.

\begin{theorem}[Uniqueness of equipotential photon surfaces]\label{MainThm}
Let $(M^{3},g)$ be a smooth Riemannian $3$-manifold with connected inner boundary $\partial M$. Suppose that $M$ carries a smooth positive lapse function $N\colon M\to\R^{+}$ such that $(M,g,N)$ solves the static vacuum equations and is asymptotically flat of mass $m\in\R$ and decay rate $\tau\geq0$. Suppose further that $\partial M$ arises as a time-slice of an equipotential photon surface in the static warped product spacetime $
(\R\times M,-N^{2}dt^{2}+g)$. Then $(M,g)$ is isometric to the exterior region of a round annulus in the Schwarzschild system $(M_{m}^{3},g_{m})$ of mass $m$, with $N$ corresponding to the according restriction of  the Schwarzschild lapse $N_{m}$ under this isometry. 

Furthermore, $\partial M$ is outward directed if and only if $m>0$, inward directed if and only if $m<0$, and degenerate if and only if $m=0$. 
\end{theorem}

\begin{coroll}[Uniqueness of photon spheres]\label{coro:sphere}
Let $(M^{3},g)$ be a smooth Riemannian $3$-manifold with connected inner boundary $\partial M$. Suppose that $M$ carries a smooth positive lapse function $N\colon M\to\R^{+}$ such that $(M,g,N)$ solves the static vacuum equations and is asymptotically flat of mass $m\in\R$ and decay rate $\tau\geq0$. Suppose further that $\partial M$ arises as a time-slice of a photon sphere  in the static warped product spacetime $
(\R\times M,-N^{2}dt^{2}+g)$. Then $(M,g)$ is isometric to a suitable piece of the Schwarzschild system $(M_{m}^{3},g_{m})$ of mass $m>0$, with $N$ corresponding to the according restriction of  the Schwarzschild lapse $N_{m}$ under this isometry.\end{coroll}

In the positive mass case, our three proofs of \Cref{MainThm} are relatively closely related to the original black hole uniqueness proofs by Israel, Robinson, and Agostiniani--Mazzieri, respectively, applying and generalizing insights and techniques from the non-degenerate photon sphere case introduced in \cite{Ced}. To treat the negative mass case requires some small changes in each of the three approaches. The reader interested in gaining more insights into the relationship between the Robinson approach and the Agostiniani--Mazzieri approach is referred to \cite[Section 6]{CCLP}.

The zero mass case can be handled by applying the positive mass theorem (see \Cref{sec:zero}) as well as by strategies somewhat related to the approaches by Israel (see \Cref{subsec:zeroIsrael}), Robinson (see \Cref{subsec:zeroRobinson}), and Agostiniani--Mazzieri (see \Cref{subsec:zeroAM}). While our Israel style approach is new, the Robinson style approach is an adaptation of the proof of the classical Willmore inequality by Cederbaum and Miehe \cite{Anabel}, and the Agostiniani--Mazzieri style approach is an adaptation of the proof of the classical Willmore inequality by Agostiniani and Mazzieri \cite{MazzAgo}, see also \cite{Ago_Fog_Maz-1}. Moreover, our Israel style proof gives a new proof of the Willmore inequality in $(\R^{3},\delta)$ under a technical assumption, see \Cref{thm:Willmore}. This suggests to view the Willmore inequality and its rigidity assertion as some kind of zero mass version or zero mass limit of the equipotential photon surface uniqueness theorems.

The uniqueness results for black holes, non-degenerate photon spheres, and outward directed equipotential photon surfaces described above have been generalized to higher dimensions~\cite{Hwang,GIS,CedGalSurface,Mazz,CCLP,Raulot}. In this work however, we will stick to $4$ spacetime dimensions as Israel's approach~\cite{israel1967event} does not seem to lend itself to a generalization to higher dimensions\footnote{On the other hand, see \cite{CCLP} in which a higher dimensional version of Robinson's approach~\cite{Rob} is presented and related to the uniqueness results in \cite{Mazz}, extending the results of the present paper to higher dimensions, both in the positive and negative mass cases. For the zero mass case, a higher dimensional analog of the zero mass case can be obtained by suitably adapting the higher dimensional proof of the Willmore inequality \cite{MazzAgo,Anabel}.} (see \Cref{rmk:high}). The described uniqueness results have also been extended to other matter models~\cite{Cedrgal2,YazaLazov,YazaLazov2,yazadjiev,CJV,BCC}. In particular, Borghini, Cederbaum, and Cogo~\cite{BCC} address connected not necessarily outward directed equipotential photon surface uniqueness in $4$-dimensional electrostatic electro-vacuum asymptotically flat spacetimes; as they assume non-vanishing electric potential, the results of \cite{BCC} do not cover the vacuum case treated here. The same applies to the work by Yazadjiev and Lazov~\cite{YazaLazov} who address sub- and super-extremal photon sphere uniqueness in the same setting as~\cite{BCC}.

\subsection*{This paper is structured as follows:}
In \Cref{sec:prelims}, we will introduce our notation and definitions and give a first proof of the zero mass case of \Cref{MainThm}. We will also list some relevant facts about non-degenerate photon spheres and equipotential photon surfaces. In \Cref{sec:zero}, we will provide similar facts about degenerate photon spheres and equipotential photon surfaces. In the brief \Cref{sec:coro}, we will deduce \Cref{coro:sphere} from \Cref{MainThm}. \Cref{sec:israel,sec:robinson,sec:AM} will be dedicated to giving proofs of \Cref{MainThm} based on the approaches by Israel, Robinson, and Agostiniani--Mazzieri, respectively. In each of these sections, we will first treat the positive/negative mass case and then treat the zero mass case.

\subsection*{Acknowledgements}
The authors would like to thank Stefano Borghini, Lorenzo Mazzieri, and David Robinson for helpful comments and questions. The work of Carla Cederbaum is supported by the focus program on Geometry at Infinity (Deutsche Forschungsgemeinschaft, SPP 2026).

\section{Preliminaries}\label{sec:prelims}
Before we recall the definitions and properties of photon spheres and equipotential photon surfaces, let us briefly collect the definitions of static vacuum spacetimes and systems (\Cref{sec:static}), introduce our asymptotic decay assumptions, and collect a few immediate consequences. As we will only address $4$-dimensional spacetimes, we restrict all definitions to these dimensions for simplicity.

All manifolds, metrics, and functions in this work will be smooth unless explicitly stated otherwise, all manifolds will be assumed to be oriented, and all submanifolds will be assumed to be embedded and oriented. Our sign and scaling convention for the mean curvature $H$ of a ($2$-)surface in a $3$-dimensional Riemannian manifold $(M^{3},g)$ is such that the unit round sphere $\mathbb{S}^{2}$ in the Euclidean space $(\R^{3},\delta)$ has mean curvature $H=2$ with respect to the unit normal $\nu$ pointing towards infinity.

\subsection{Static vacuum spacetimes and systems}\label{sec:static}
In this paper, a \emph{(standard) static spacetime} is a $4$-dimensional manifold of the form $\R\times M^{3}$ for some $3$-dimensional manifold $M$, carrying a Lorentzian metric of the form
\begin{align}\label{eq:Lorentzian}
-N^{2}dt^{2}+g,
\end{align}
where $N\colon M\to\R^{+}$ is a positive function called the \emph{(static) lapse function}, $g$ is a Riemannian metric on $M$, and $t\in\R$. The same information can alternatively be encoded in the \emph{static system} $(M^{3},g,N)$. If $M$ has a boundary $\partial M$, we will assume that both $N$ and $g$ smoothly extend to $\partial M$, with $N>0$ on $\partial M$. Consequently,  \eqref{eq:Lorentzian} smoothly extends to $\R\times\partial M$. To allow for (non-warped product) photon surfaces to arise as the inner boundary of a static spacetime, we will slightly abuse notation and call a $4$-dimensional Lorentzian manifold $(\mathcal{L}^{4},\mathfrak{g})$ a \emph{static spacetime} if it is the closure of an open subset of a (standard) static spacetime $(\R\times M^{3},-N^{2}dt^{2}+g)$, with ``inner'' boundary $\partial\mathcal{L}\subset\R\times\left(M\cup\partial M\right)$.

Next, a static system $(M^{3},g,N)$ (and the associated static spacetime) are called \emph{vacuum} or said to satisfy the \emph{static vacuum equations} if 
\begin{align}\label{eq:SVE1}
N \Ric &= \nabla^2N,\\\label{eq:SVE2}
\Delta N &= 0 
\end{align}
hold on $M$, where $\Ric$, $\nabla$, $\nabla^{2}$, and $\Delta$ denote the Ricci tensor, Levi-Civita connection, Hessian, and Laplace--Beltrami operator with respect to $g$, respectively. As is well-known, the scalar curvature $\Scal$ of a static vacuum system $(M,g,N)$ vanishes by a combination of \eqref{eq:SVE1} and \eqref{eq:SVE2}, i.e.,
\begin{align}\label{scal0}
\Scal=0
\end{align}
holds on $M$. The undisputedly most important example of a static vacuum system is the \emph{Schwarz\-schild system $(M^{3}_{m},g_{m},N_{m})$ of mass $m\in\R$}, given by 
\begin{align}
\begin{split}\label{Schw}
    N_{m}(r) &= \sqrt{1-\frac{2m}{r}},\\
    g_{m}&=\frac{dr^{2}}{N_{m}(r)^{2}}+r^{2} g_{\sphere},
\end{split}
\end{align}
on $M^3_m= (2m,\infty) \times \sphere$ when $m>0$ and on $M^{3}_{m}=(0,\infty)\times\sphere$ when $m\leq0$, where $ g_{\sphere}$ denotes the canonical metric on $\sphere$. For $m=0$, the Schwarzschild system corresponds to Euclidean space with the origin removed and written in spherical polar coordinates; accordingly, the corresponding spacetime is isometric to the Minkowski spacetime away from the worldline of the coordinate origin. For $m>0$, the Schwarzschild system or suitable subsets thereof model the exterior region of a static, spherically black hole or star at one instant of ``static time $t$''. For $m<0$, the Schwarzschild system and the corresponding spacetime become singular as $r\to0+$ and are considered unphysical.

A static system $(M^{3},g,N)$ will be called \emph{asymptotically flat of mass $m\in\R$} if, outside some compact set $K\subset M$, it is diffeomorphic to the exterior region of a closed ball $B\subset\R^{3}$, $M\setminus K\approx \R^{3}\setminus B$, and if, in the coordinates $(x^{i})$ induced by this diffeomorphism, we have
\begin{align}\label{eq:asyg}
g_{ij}&=\delta_{ij}+o_{2}(\vert x\vert^{-\tau})\\\label{eq:asyN}
N &= 1 - \frac{m}{|x|}  + o_2(|x|^{-1})  \text{ as } |x| \rightarrow +\infty
\end{align}
as $|x| \to \infty$ in the coordinates induced by this diffeomorphism, with \emph{decay rate} $\tau\geq0$. Here, for a given smooth function $f \colon \R^{3} \rightarrow \R$, the notation $f = o_l (|x|^{\alpha})$ as $\vert x\vert\to\infty$ for some $l \in \mathbb{N}$, $\alpha \in \R$ is an abbreviation for
\begin{equation*}
 \sum_{|J|\leq l} |x|^{\alpha+|J|}|\partial^J f| = o(1)
\end{equation*}
as $|x|\rightarrow \infty$, where $J$ runs through all multi-indices with $\vert J\vert\leq l$. In this paper, we will not give many details about asymptotic considerations as these are all rather standard in the $3$-dimensional setting. We refer the interested reader e.g.\ to \cite[Lemma 2.5]{CCLP} for more details. In particular, we note that the mass parameter $m$ is geometric, i.e., independent of the choice of asymptotic coordinates.

Before we move on, let us note that asymptotically flat static systems $(M^{3},g,N)$ are necessarily metrically and geodesically complete (up to the boundary $\partial M$) and such that the necessarily at most finitely many connected components of $\partial M$ are all closed, see for example \cite[Appendix]{CGM}. Here, being \emph{geodesically complete up to the boundary} means that any geodesic $\gamma\colon I\to M$ which is not defined on all of $\R$, $I\neq\R$, can be smoothly extended to a geodesic $\widehat{\gamma}\colon J\to M\cup\partial M$, $J\supseteq I$, such that either $J=\R$, $J=[a,\infty)$, $J=(-\infty,b]$, or $J=[a,b]$ for some $a,b\in\R$ such that $\widehat{\gamma}(a), \widehat{\gamma}(b)\in\partial M$ (whenever applicable).

Now let $(M^{3},g,N)$ be an asymptotically flat static vacuum system with connected boundary $\partial M$ and mass $m$. We set
\begin{equation} \label{eq: israel_lambda}
    \lambda \definedas  \sign (m)
\end{equation}
for convenience, with the convention that $\sign(0)\definedas0$. For any regular value $s>0$ of $N$, let $\nu_{s}$ denote the unit normal to the (regular) level set $\Sigma^{2}_{s}\definedas \{N=s\}$ of $N$. Then by the Hopf lemma and the maximum principle\footnote{which can both be modified to allow for non-compact $M$ under our asymptotic flatness assumptions \eqref{eq:asyg}, \eqref{eq:asyN}; we implicitly use here that $(M,g)$ is complete up to the boundary.} applied to the static vacuum equation~\eqref{eq:SVE2}, the normal derivative $\nu_{s}(N)$ of $N$ on any regular level set $\Sigma^{2}_{s}$ of $N$ satisfies
\begin{align}\label{eq:signspre}
\sign(\nu_{s}(N))=\sign(1-s).
\end{align}
Equivalently, provided that $s$ is a regular value of $N$, one has
\begin{align}\label{eq:normals}
\nu_{s}=\sign(1-s)\frac{\nabla N}{\vert\nabla N\vert}
\end{align}
on $\Sigma^{2}_{s}$, where $\nabla N$ denotes the gradient of $N$ and $\vert\cdot\vert$ denotes the pointwise tensor norm with respect to $g$. Moreover, the \emph{Smarr formula}
\begin{align}\label{eq:Smarr}
\int_{\Sigma^{2}_{s}}\nu_{s}(N)\,dA=4\pi m
\end{align}
holds for all regular values $s$ of $N$ by a simple application of the divergence theorem to \eqref{eq:SVE2}, exploiting the geodesic completeness of $(M,g)$ up to the boundary and our asymptotic assumptions \eqref{eq:asyg}, \eqref{eq:asyN}. Here, $dA$ denotes the area element induced on $\Sigma^{2}_{s}$ by $g$. By \eqref{eq:signspre}, this implies 
\begin{align}\label{eq:signs}
\sign(\nu_{s}(N))=\sign(1-s)=\sign(m)=\lambda
\end{align}
for all regular values $s$ of $N$. We can summarize this as follows.

\begin{prop}[Sign of mass $m$]\label{prop:signm}
Let $(M^{3},g,N)$ be an asymptotically flat static vacuum system with connected boundary $\partial M$, mass $m$, and decay rate $\tau\geq0$. Suppose that $N\vert_{\partial M}=N_{0}$ for some constant $N_{0}>0$ and let $\nu$ denote the unit normal to $\partial M$ pointing towards infinity. Then 
\begin{align}\label{eq:signbdry}
\sign(1-N_{0})&=\sign(m)=\lambda.
\end{align}
Furthermore, if $N_{0}\neq1$, $\partial M$ is a regular level set of $N$, $\partial M=\Sigma^{2}_{N_{0}}$, with
\begin{align}\label{eq:normal}
\nu&=\lambda \frac{\nabla N}{\vert\nabla N\vert}
\end{align} 
on $\partial M$.
\end{prop}
\begin{proof}
First suppose that $N_{0}\neq1$ and note that by the Hopf lemma, $\partial M$ is a regular level set of $N$. The claim in \eqref{eq:signbdry} then follows readily from \eqref{eq:signs} while \eqref{eq:normal} follows from \eqref{eq:normals}. Second, suppose that $N_{0}=1$. Then as $N\to1$ asymptotically, the maximum principle applied to \eqref{eq:SVE2} implies that $N\equiv1$ and hence by definition of the asymptotic mass in \eqref{eq:asyN} we have $m=0$ and thus $\lambda=0$ which gives \eqref{eq:signbdry}.
\end{proof}

\subsection{Photon spheres and equipotential photon surfaces}\label{sec:photo}
Let us now recall the basic definitions and properties of equipotential photon surfaces and of photon spheres. For more detailed information, we refer the interested reader to \cite{CedGalSurface,CJV} and the references given therein. First, a timelike hypersurface $P^3\subset\mathfrak{L}^{4}$ in a smooth Lorentzian manifold $(\mathcal{L}^{4},\mathfrak{g})$ is called a \emph{photon surface} if any null geodesic initially tangent to $P$ remains tangent to $P$ as long as it exists or in other words if $P$ is \emph{null totally geodesic} in $(\mathcal{L}^{4},\mathfrak{g})$. Equivalently, $P$ is a photon surface if it is totally umbilic in $(\mathcal{L},\mathfrak{g})$. If the ambient spacetime is (standard) static with $(\mathcal{L}\subseteq\R\times M^{3},\mathfrak{g}=-N^{2}dt^{2}+g)$ then $P$ is \emph{equipotential} if the lapse function $N$ is constant along each connected component of each \emph{time-slice} $\Sigma^{2}(t) \definedas  P \cap \left(\{t\}\times M\right)$ of $P$. If $N\equiv\text{const.}$ along $P$ and $P=\R\times\Sigma^{2}$ for some connected surface $\Sigma^{2}$ then $P$ is called a \emph{photon sphere} (regardless of the topology of $\Sigma$).

If $P^{3}$ is a photon surface in a static spacetime with lapse $N$ then $P$ is called \emph{non-degenerate} if $\nabla N(x)\neq0$ for all $x\in P$, otherwise it is called \emph{degenerate}. If the underlying static system $(M^{3},g,N)$ is asymptotically flat and $\eta$ denotes the unit normal to $P$ pointing to infinity, then $P$ is called \emph{outward} respectively \emph{inward directed} if $\eta(N)>0$ respectively $\eta(N)<0$ on $P$. 

It is established in \cite{CedGalSurface,CJV} that each \emph{time-slice $\Sigma^{2}=\Sigma^{2}(t)$} of a non-degenerate equipotential photon surface in a static vacuum spacetimes is umbilic, $\mathring{h}=0$, has constant mean curvature $H_{0}\neq0$, constant normal derivative of the lapse $\nu(N)$, and constant scalar curvature $\Scal_{\Sigma}$, satisfying the \emph{photon surface constraints}
 \begin{align}\label{normalLapse}
2\nu(N) &= c N_{0}H_{0} ,\\\label{scalarMeanCurv}
 \Scal_{\Sigma} &= \left( c + \tfrac{1}{2} \right) H^2_{0}
 \end{align}
 for some constant $c\neq0$, where $N_{0}\definedas N\vert_{\partial M}>0$. Moreover, $H_{0}>0$ holds with respect to the outward pointing unit normal $\nu$ to $\Sigma^{2}$ if the ambient static vacuum spacetime is asymptotically flat with $\tau\geq0$. This is shown\footnote{In fact, both of these sources make stronger asymptotic decay assumptions than we do and, in addition, assume outward directedness, $\nu(N)>0$, and $H_{0}\nu(N)>0$ on $\partial M$, respectively. As $\partial M$ is connected in our setting, neither of the second assumptions is needed to conclude as one sees in the corresponding proofs. This is because these conditions are only needed to handle potential other boundary components. Moreover, our asymptotic decay assumptions imply that large coordinate spheres have positive mean curvature which is the other ingredient needed to conclude in these sources.} in \cite[Lemma 2.6]{CedGalSurface}, \cite[Theorem 5.22]{CJV}.
 
 Specifically, $c=1$ holds for non-degenerate photon spheres. The photon surface constraints relate the topology of $\Sigma$ and the constant $c$ via Gau\ss--Bonnet theorem, i.e.,
     \begin{equation} \label{eq: Gauss-Bonnet for Sigma}
         4 \pi \chi (\Sigma) = \int_{\Sigma} \left( c + \frac{1}{2} \right) H^2_{0}\,dA = \left( c+ \frac{1}{2} \right) H^2_{0}  \abs{\Sigma},
     \end{equation}
     where $\chi (\Sigma)$ is the Euler characteristic of $\Sigma$. In particular, $c > -\tfrac{1}{2}$ holds if and only if $\Sigma$ is topologically a sphere (i.e., $\chi (\Sigma) = 2$).
    
As discussed in the introduction, the Schwarzschild spacetimes host a zoo of equipotential photon surfaces which are all spherically symmetric. If $m\leq0$, the Schwarzschild spacetimes host no photon spheres. If $m>0$, they host a unique photon sphere at $r=3m$. One can check e.g.\ with the help of \Cref{prop:signm} that all equipotential photon surfaces in a positive mass Schwarzschild spacetime are outward directed while those in a negative mass Schwarzschild spacetime are inward directed. All photon surfaces in the zero mass Schwarzschild (aka Minkowski) spacetime are degenerate. One easily computes that $\sign(c)=\sign(m)$ for equipotential photon surfaces in the Schwarzschild spacetimes of mass $m$, with all above-mentioned properties of non-degenerate equipotential photon surfaces, including \eqref{normalLapse}, \eqref{scalarMeanCurv} readily extending to spherically symmetric photon surfaces in the zero mass Schwarzschild (aka Minkowski) spacetime, with $c=0$. Moreover, $c>-\frac{1}{2}$ in negative mass Schwarzschild spacetimes readily follows from spherical symmetry of equipotential photon surfaces.

\section{The degenerate equipotential photon surface or zero mass case}\label{sec:zero}
Let us now discuss the case of a degenerate equipotential photon surface. To this end, let $(M^{3},g,N)$ be an asymptotically flat static vacuum system with connected boundary $\partial M$, mass $m$, and decay rate $\tau\geq0$. Assume that $\partial M$ arises as a time-slice of a degenerate equipotential photon surface $P^{3}\subset\R\times M$ with at least one of the degenerate points lying on $\partial M$ (otherwise switch to a different time-slice of the photon surface). By the Hopf lemma and the maximum principle, the existence of one degenerate point on $\partial M$ implies that $N$ is constant on $M$, thus by our asymptotic assumption \eqref{eq:asyN}, $N\equiv1$ on $M$ as claimed and hence by definition $m=0$. Consequently, the entire equipotential photon surface with time-slice $\partial M$ is everywhere degenerate, giving $\nabla N\equiv0$ on $P$.

Next, by the static vacuum equation \eqref{eq:SVE1}, $N\equiv1$ gives $\Ric=0$ so that $(M^{3},g)$ is necessarily flat. It is well-known that each connected photon surface in the Minkowski spacetime arises as a subset of a timelike hyperplane or of a rotationally symmetric one-sheeted hyperboloid. The arguments to show this are completely local, hence the same applies to suitably small connected subsets of photon surfaces in our spacetime $(\R\times M^{3},-dt^{2}+g)$, in the sense of their intrinsic and extrinsic properties. As the photon surface we are investigating is connected with closed, connected time-slices, we can exclude the hyperplane case even locally by continuity of the (constant) mean curvature and conclude that $\partial M$ is necessarily intrinsically and extrinsically a round sphere of some radius $r_{0}>0$. This establishes the following lemma.

\begin{lemma}[Properties of degenerate equipotential photon surfaces]\label{lem:zero}
Let $(M^{3},g,N)$ be an asymptotically flat static vacuum system with connected boundary $\partial M$, mass $m$, and decay rate $\tau\geq0$. Assume that $\partial M$ arises as a time-slice of a degenerate equipotential photon surface $P^{3}\subset\R\times M$. Then $m=0$, $N\equiv1$, $(M,g)$ is flat, and $\partial M$ is isometric to the round sphere $\sphere_{r_{0}}$ for some radius $r_{0}$, is umbilic, and has constant mean curvature $H_{0}=\frac{2}{r_{0}}$ in $(M,g)$. In particular, the photon surface constraints \eqref{normalLapse}, \eqref{scalarMeanCurv} are satisfied on $\partial M$ for $c=0$. Moreover, for each point in $\partial M$, there is a one-sided tubular neighborhood isometric to a one-sided tubular neighborhood of $\sphere_{r_{0}}$ in $(\R^{3},\delta)$.
\end{lemma}

Next, we will give our first proof of the zero mass case of \Cref{MainThm} worded as \Cref{thm:zero} below, based on the rigidity assertion of the Riemannian positive mass theorem \cite{SchoenYau,Witten} which applies to geodesically complete asymptotically flat Riemannian $3$-manifolds with decay rate $\tau>\tfrac{1}{2}$. In some sense, this can be considered as an analogous proof to the uniqueness proofs for non-degenerate photon spheres given in \cite{carlagregpmt} and for outward directed equipotential photon surfaces given in \cite{CedGalSurface}, except that we restrict to connected $\partial M$ and that we exploit the vanishing of $m$ to bootstrap our decay from $\tau>0$ to $\widehat{\tau}>\tfrac{1}{2}$ by switching into harmonic coordinates (which is weaker than the asymptotically Schwarzschildean decay assumed in \cite{carlagregpmt,CedGalSurface}). We will comment on the case $\tau=0$ in \Cref{rmk:AFM}.

\begin{theorem}[Degenerate case]\label{thm:zero}
Let $(M^{3},g,N)$ be an asymptotically flat static vacuum system with connected boundary $\partial M$, mass $m$, and decay rate $\tau>0$. Assume that $\partial M$ arises as a time-slice of a degenerate equipotential photon surface $P^{3}\subset\R\times M$. Then $(M,g)$ is isometric to a suitable piece of $(\R^{3},\delta)$ and $N\equiv1$ on $M$, as well as $m=0$. Moreover, the entire equipotential photon surface with time-slice $\partial M$ is everywhere degenerate.
\end{theorem}
\begin{proof}
We have already established $N\equiv1$ on $M$ and $m=0$ as well as $\Ric=0$  on $M$  in \Cref{lem:zero}. Before we move on, we note that one can bootstrap the decay assumption $\tau>0$ to $\widehat{\tau}>\frac{1}{2}$ by switching to harmonic coordinates and exploiting $\Ric=0$: First, we note that our assumption \eqref{eq:asyg} implies that $(M,g)$ is indeed asymptotically flat with respect to the same asymptotic coordinates $(x^i)$ in the sense of weighted Sobolev spaces for any decay rate $0<\tau^*<\tau$ and all exponents $3< q < \infty$, that is, 
\begin{equation}\label{a priori decay}
g_{ij} - \delta_{ij} \in W^{2, q}_{-\tau^*}(\R^3\setminus \overline{B})
\end{equation}
with respect to $(x^i)$. From Bartnik's result  \cite[Theorem 3.1]{Bartnik}, we thus know that there exist harmonic coordinates $(y^i)$ (on a possibly smaller asymptotic end of $M$) such that \eqref{a priori decay} also holds with respect to $(y^i)$ (but possibly outside a different ball $\widetilde{B}$). We can now apply Bartnik's result \cite[Proposition 3.3]{Bartnik} to deduce that, for every $0<\widehat{\tau}<1$ and every $3<q<\infty$, we indeed have 
\begin{equation}
g_{ij} - \delta_{ij} \in W^{2, q}_{-\widehat{\tau}}(\R^3\setminus \overline{\widetilde{B}})
\end{equation}
with respect to $(y^i)$. In particular, the asymptotic assumptions of the positive mass theorem hold with respect to $(y^i)$.

To see that $(M^{3},g)$ is globally isometric to a piece of $(\R^{3},\delta)$ by virtue of the rigidity case of the Riemannian positive mass theorem, glue in a round ball of radius $r_{0}$ into $\partial M$ to smoothly extend $(M,g)$ to a smooth\footnote{This follows in Riemann normal coordinates from flatness near $\d M$.}, geodesically complete, asymptotically flat manifold $(\overline{M}^{3},\overline{g})$ with decay rate $\widehat{\tau}>\frac{1}{2}$ by \Cref{lem:zero} and the assertions on the local geometry of time-slices of degenerate equipotential photon surfaces above \Cref{lem:zero}. By the work of Herzlich~\cite[Theorem 2.3]{Herzlich}, we know that the ADM mass of $(\overline{M},\overline{g})$ vanishes by (Ricci) flatness of $(M,g)$. Hence by the rigidity case of the Riemannian positive mass theorem, $(\overline{M},\overline{g})$ is globally isometric to Euclidean space which shows that $(M^{3},g)$ is isometric to a suitable piece of $(\R^{3},\delta)$ as claimed. 
\end{proof}

Moreover, we can see that the zero mass case $m=0$ and the degenerate equipotential photon surface case are actually identical, as stated in the following remark.

\begin{rmk}[Zero mass versus degenerate]\label{rem:zerodegen}
Let $(M^{3},g,N)$ be an asymptotically flat static vacuum system with connected boundary $\partial M$, mass $m$, and decay rate $\tau\geq0$. Assume that $\partial M$ arises as a time-slice of an equipotential photon surface $P^{3}\subset\R\times M$. Then $m=0$ if and only if $P$ is degenerate. 
\end{rmk}
\begin{proof}
If $m=0$, the Smarr formula \eqref{eq:Smarr} implies that $\int_{\partial M}\nu(N)\,dA=0$. By the Hopf lemma, either $N\equiv\text{const}$ or $\nu(N)$ has a sign on $\partial M$ with the latter contradicting this integral inequality; thus $N\equiv1$ on $M$ by our asymptotic assumption \eqref{eq:asyN}. Consequently, $\nabla N\equiv0$ on $M$ and $P$ is degenerate. If, on the other hand, $P$ is degenerate then by \Cref{lem:zero} we know that $m=0$.
\end{proof}

\begin{rmk}[Alternative proof with weaker assumptions on the decay rate]\label{rmk:AFM}
Alternatively, one can argue via the Willmore-type inequality proved by Agostiniani--Fogagnolo--Mazzieri~\cite[Theorem 1.1]{Ago_Fog_Maz-1} also for $\tau=0$. To see this, let us consider the geodesically complete Riemannian manifold $(\overline{M}^{3},\overline{g})$ constructed in the proof of \Cref{thm:zero} and recall that it is smooth and Ricci flat, hence satisfies $\overline{\Ric}\geq0$. Also, by the asymptotic flatness condition and via work of Borghini and Fogagnolo~\cite{BF}, the asymptotic volume ratio satisfies $\operatorname{AVR}(\overline{g})=\operatorname{AVR}(g)=1$. Hence the Willmore-type inequality holds for the domain $\Omega$ coinciding with the flat round ball we glued in in the proof of \Cref{thm:zero}, and indeed with equality by \Cref{lem:zero} as $\partial\Omega=\partial M$. As the induced metric on $\partial\Omega=\partial M$ is round, we conclude from the rigidity assertion of \cite[Theorem 1.1]{Ago_Fog_Maz-1} that $(M,g)$ is isometric to the exterior region of a round ball of radius $r_{0}$ in $(\R^{3},\delta)$.
\end{rmk}

Before we move on, let us prepare the stage for the three other proofs we will give of \Cref{thm:zero}. In the zero mass (or degenerate equipotential photon surface) case (see \Cref{rem:zerodegen}), we know from the above that $N\equiv1$ on $M$ and hence a level set approach via $N$ as used by the strategies based on the approaches by Israel and Agostiniani--Mazzieri and also indirectly as well as in the rigidity analysis of the strategy based on the approach by Robinson does not work. However, a level set approach works if one instead uses the \emph{electrostatic potential}, that is, the unique smooth function $u\colon M\to\R$ satisfying
\begin{align}\label{eq:harmonic}
\Delta u=0
\end{align}
on $M$, the boundary condition
\begin{align}\label{eq:bdryu}
u\vert_{\partial M}=1,
\end{align}
and the asymptotic condition $u\to0$ as $\vert x\vert\to\infty$ in the asymptotic end of $(M,g,N)$. This unique solution $u$ exists by standard elliptic theory (see e.g.\ \cite[Chapters 2, 3]{Gilbarg.2001}), noting that the Sobolev spaces with respect to $(M,g)$ are "asymptotically equivalent" to those on $(\R^3,\delta)$. If $(M,g)$ is isometric to $(\R^{3}\setminus\overline{\Omega},\delta)$ for some bounded domain $\Omega$ with smooth boundary, it is well-known that the unique solution decays as
\begin{align}\label{eq:asyu}
u=\frac{R_{0}}{\vert y\vert}+o_{2}(\vert y\vert^{-1})
\end{align}
for some constant $R_{0}\in\R$ as $\vert y\vert\to\infty$ (see e.g.\ \cite{Kellogg.1967}), with respect to the Cartesian coordinates $(y^{i})$ on $\R^{3}$ pulled back to $M$. In our setting, we only know that $(M,g)$ is asymptotically flat with decay rate $\tau\geq0$ and flat, but not (to the best knowledge of the authors) whether there exist Cartesian\footnote{Here, by \emph{Cartesian coordinates}, we mean coordinates with respect to which the flat metric appears as $\delta_{ij}$.} coordinates $(y^{i})$ in some exterior region $M\setminus K$ diffeomorphic to the exterior region of a closed ball in Euclidean space, for some $K\subset M$ compact. If such coordinates exist and relate to the originally given asymptotically flat coordinates $(x^{i})$ by a transformation asymptotically satisfying $y^{i}=x^{i}+o_{3}(\vert x\vert^{1-\tau})$ as $\vert x\vert\to\infty$, we have \eqref{eq:asyu} also with respect to $(x^{i})$. We will assume for simplicity throughout this paper that \eqref{eq:asyu} applies with respect to the given asymptotically flat coordinates $(x^{i})$. One could possibly circumvent this problem by arguments along the line of those given in \cite{Ago_Fog_Maz-1,BF} or exploiting Bartnik style results similar to those used in the proof of \Cref{thm:zero} but this would lead too far away from the focus of this work.

Under the assumption that \eqref{eq:asyu} holds with respect to $(x^{i})$ and arguing as for the Smarr formula \eqref{eq:Smarr}, one finds that
\begin{align}\label{eq:Smarrb}
\int_{\Sigma^{2}_{s}}\nu_{s}(u)\,dA=-4\pi R_{0}
\end{align}
on regular level set $\Sigma^{2}_{s}\definedas \{u=s\}$ of $u$, where $dA_{s}$ denotes the area measure induced on $\Sigma_{s}$ by $g$. Moreover, by the maximum principle, we know that $0<u<1$ on $M$ and thus $b\geq0$. By the Hopf lemma, we know that $\nu(u)<0$ on $\partial M$ so that in particular $R_0>0$ by \eqref{eq:Smarrb} as $\partial M$ is a regular level set of $u$.

\section{Deducing \Cref{coro:sphere} from \Cref{MainThm}}\label{sec:coro}
Before we move on to proving \Cref{MainThm}, let us demonstrate how \Cref{MainThm} implies \Cref{coro:sphere}. First, applying \Cref{MainThm} to a non-degenerate photon sphere gives the desired claim as non-degenerate photon spheres are necessarily outward directed by \cite[Lemma 2.6, Remark 2.7]{carlagregpmt} and thus $m>0$. Now suppose towards a contradiction that there exists a degenerate photon sphere. By \Cref{lem:zero} and the considerations above it, the photon sphere is locally embedded in the Minkowski spacetime -- which we know not to host any photon spheres (even locally) other than pieces of flat timelike hyperplanes. On the other hand, we know that the photon sphere is locally isometric to a straight cylinder over a round sphere, hence cannot be part of a flat timelike hyperplane. This excludes the possibility of a degenerate photon sphere and finishes the proof that \Cref{MainThm} implies \Cref{coro:sphere}. We have also proved the following previously unknown result.

\begin{coroll}[No connected degenerate photon spheres]
Let $(M^{3},g)$ be a Riemannian $3$-manifold with connected inner boundary $\partial M$. Suppose that $M$ carries a positive lapse function $N\colon M\to\R^{+}$ such that $(M,g,N)$ solves the static vacuum equations and is asymptotically flat of mass $m\in\R$ and decay rate $\tau\geq0$. Suppose further that $\R\times\partial M$ is  a photon sphere in $
(\R\times M,-N^{2}dt^{2}+g)$. Then this photon sphere is necessarily non-degenerate and $m\neq0$.
\end{coroll}

\begin{rmk}[No degenerate photon spheres whatsoever]
In fact, arguing as above via the Hopf lemma and local arguments, one can see that a static vacuum asymptotically flat static spacetime with inner boundary consisting of (individually) degenerate photon spheres cannot exist.
\end{rmk}

\section{Photon Surface Uniqueness  \`a la Israel} \label{sec:israel}
\subsection{The positive and negative mass cases  \`a la Israel}\label{subsec:nonzero}
In \cite{Ced}, Cederbaum adapted Israel's proof of static vacuum black hole uniqueness~\cite{israel1967event} to show that the Schwarzschild spacetimes of positive mass are the only static vacuum asymptotically flat spacetimes that possess a connected non-degenerate photon sphere inner boundary, assuming asymptotic spherical symmetry and the technical condition $\nabla N\neq0$ on $M$ (as does Israel). We will now generalize this procedure to the non-degenerate equipotential photon surface case and to our much weaker decay assumptions \eqref{eq:asyg}, \eqref{eq:asyN} for decay rate $\tau\geq0$, see \Cref{thm:Israel} below. As we will see, this will give a proof of \Cref{MainThm} in the positive and negative mass cases under the additional assumption $\nabla N\neq0$ on $M$. Our proof reduces to that given in \cite{Ced} when $c=1$ and thus in particular in the photon sphere case.

\begin{theorem}[Non-degenerate case \`a la Israel]\label{thm:Israel}
Let $(M^{3},g,N)$ be an asymptotically flat static vacuum system of mass $m\in\R$ and decay rate $\tau\geq0$ with connected inner boundary $\partial M$ and assume that $\nabla N\neq0$ on $M$. Suppose further that $\partial M$ arises as a time-slice of a non-degenerate equipotential photon surface. Then $(M,g)$ is isometric to the exterior region of a round annulus in the Schwarzschild system $(M_{m}^{3},g_{m})$ of mass $m$, with $N$ corresponding to the according restriction of the Schwarzschild lapse $N_{m}$ under this isometry. Furthermore $m\neq0$, $c > -\tfrac{1}{2}$ and $N_{0}>1$ if $m>0$ while $0<N_{0}<1$ if $m<0$, where $N_{0}\definedas N\vert_{\partial M}$.
\end{theorem}

\begin{proof}
First, observe that the assumption $\nabla N\neq0$ on $M$ guarantees that $N$ regularly foliates $M$. In particular, we know that $N_{0}\neq1$ by the maximum principle. Hence we can use $N$ as a global coordinate function on $M$ and construct a local coordinate system by flowing local coordinates $(y^{I})$, $I=1,2$, on $\partial M$ along the nowhere vanishing vector field $\frac{\nabla N}{\vert \nabla N\vert^{2}}$ when $N_{0}>1$ and along $-\frac{\nabla N}{\vert \nabla N\vert^{2}}$ when $0<N_{0}<1$, recalling that $N_{0}>0$ by assumption and thus $N>0$ by the maximum principle. In these coordinates, the metric $g$ takes the form
\begin{align*}
g=\rho^{2}dN^{2}+\sigma,
\end{align*}
where $\sigma=\sigma(N)$ denotes the Riemannian metric induced on the regular level set $\Sigma_{N}$ of $N$ and $\rho \definedas  \frac{1}{|\nabla N|}$. Next, from \Cref{rem:zerodegen}, we know that $m\neq0$ by the non-degeneracy assumption on $\partial M$. Thus, the asymptotic assumption \eqref{eq:asyN} guarantees that large level sets $\Sigma_{N}$ of $N$ are topological spheres. As $N$ regularly foliates $M$, all level sets $\Sigma_{N}$ have the same topology, hence they are all topological spheres, including $\partial M$. Let 
\begin{align}\label{def:arearadius}
r(N) \definedas  \sqrt{\frac{|\Sigma_{N}|}{|\mathbb{S}^2|_{ g_{\sphere}}}}
\end{align}
denote the \emph{area radius} of $\Sigma_{N}$. Set $r_{0}\definedas r(N_{0})$. Using this, the properties of equipotential photon surfaces and in particular the photon surface constraint \eqref{normalLapse}, and the Smarr formula \eqref{eq:Smarr}, we find
\begin{equation}\label{NewEqc}
\frac{cN_0 H_{0}}{2} =  \frac{m}{r_0^2}.
\end{equation}
From the photon surface constraint \eqref{scalarMeanCurv} and the Gau{\ss}--Bonnet theorem combined with the fact that $\partial M$ is a topological sphere, we get
\begin{equation*}
  \left(c + \tfrac{1}{2}\right)  H_{0}^2 r_0^2 = 2
\end{equation*}
which, replacing $c$ via \eqref{NewEqc}, can be reformulated as
\begin{equation}\label{eq: israel_mc2}
\frac{4mH_{0}}{N_{0}} +  H^2_{0} r_0^2 = 4.
\end{equation}

In particular, as $H_{0}>0$, it follows that $c > -\tfrac{1}{2}$ as claimed. Recall from \eqref{eq:signs} that $\lambda=\sign(m)=\sign(1-N_{0})=\sign(\nu(N)\vert_{\partial N})$. As $m\neq0$, this proves the claims on the signs of $m$ and the size of $N_{0}$. Using the same coordinates and only local arguments,  it is shown in \cite{Ced} that the static vacuum equations \eqref{eq:SVE1}, \eqref{eq:SVE2} imply
\begin{align}\label{Is1}
    0 & = \frac{\lambda}{\rho} \left( \frac{H}{N} - H,_N - \frac{\lambda \rho}{2} H^2 \right) - \frac{2}{\sqrt{\rho}}\,\Delta_{\sigma} \sqrt{\rho} - \frac{1}{2}\left[ \frac{|\nabla_{\!\sigma} \rho|_{\sigma}^2}{\rho^2} + 2|\mathring{h}|_{\sigma}^2 \right], \\ \label{Is2}
    0 & = \frac{\lambda}{\rho} \left( \frac{3H}{N} - H,_N \right) - \Scal_{\Sigma} - \Delta_{\sigma} \log \rho -\left[ \frac{|\nabla_{\!\sigma} \rho|_{\sigma}^2}{\rho^2} + 2|\mathring{h}|_{\sigma}^2 \right], \\ \label{Is3}
    0 & = \rho,_N-\lambda \rho^2 H,
\end{align}
on each level set $\Sigma_{N}$, where $\Delta_{\sigma}$, $\nabla_{\!\sigma}$, and $\vert\cdot\vert_{\sigma}$ denote the Laplace--Beltrami operator, covariant gradient, and pointwise tensor norm with respect to $\sigma$, respectively, and $\mathring{h}$ denotes the trace-free part of the second fundamental form of $\Sigma_{N}$. Using \eqref{Is3} and the non-negativity of the terms in the square brackets, \eqref{Is1} and \eqref{Is2} give
\begin{align}\label{IsIn00}
\partial_N \left( \frac{\lambda H \sqrt{\mathfrak{s}}}{\sqrt{\rho}N} \right) & \leq - \frac{2}{N} \Delta_{\sigma} \sqrt{\rho} \,\sqrt{\mathfrak{s}},
\\ \label{IsIn01}
\partial_N \left( \frac{\lambda}{\rho}\left[ H N + \frac{4\lambda}{\rho} \right] \sqrt{\mathfrak{s}}\right) & \leq -N\left(\Scal_{\Sigma} +\Delta_{\sigma} \log \rho\right)\sqrt{\mathfrak{s}} ,
\end{align}
on each level set $\Sigma_{N}$, where $\mathfrak{s}\definedas\operatorname{det}(\sigma_{IJ})$ and $\partial_{N}\sqrt{\mathfrak{s}}=\lambda\rho H\sqrt{\mathfrak{s}}$. In these inequalities\footnote{Note that \cite[Equation 3.27]{Ced} contains a typo, namely it misses the factor $\lambda$ on the left hand side appearing in \eqref{IsIn01}.}, equality holds if and only if the terms in the square brackets of \eqref{Is1}, \eqref{Is2} both vanish. Integrating \eqref{IsIn00} from $N_{0}$ to $1$ (recalling that $N_{0}\neq1$ by the above) and subsequently over the domain of definition $U$ of the coordinates $(y^{I})$, one obtains
\begin{align*}
\lim_{t\to1}\int_{\Sigma_{t}}\frac{\lambda H}{\sqrt{\rho}N}\,dA-\int_{\partial M} \frac{\lambda H}{\sqrt{\rho}N}\,dA&\leq-2\int_{N_{0}}^{1}\!\frac{1}{N}\! \int_{\Sigma_{N}}\Delta_{\sigma} \sqrt{\rho} \,dA\, dN=0,\\
\lim_{t\to1}\int_{\Sigma_{t}} \!\frac{\lambda}{\rho}\left[ H N + \frac{4\lambda}{\rho} \right]dA-\!\int_{\partial M}\! \frac{\lambda}{\rho}\left[ H N + \frac{4\lambda}{\rho} \right]dA&\leq-\int_{N_{0}}^{1}\!\!N\!\!\int_{\Sigma_{N}}\left( \Scal_{\Sigma} +\Delta_{\sigma} \log \rho \right)dA\,dN\\
&=-8\pi\int_{N_{0}}^{1}N\,dN=-4\pi\left(1-N_{0}^{2}\right)
\end{align*}
where we have used the fundamental theorem of calculus, Fubini's theorem, a partition of unity, the divergence theorem, the Gauss--Bonnet theorem together with our knowledge that all $\Sigma_{N}$ are topological spheres, and where $dA$ denotes the area element induced by $\sigma$ on $\Sigma_{N}$. Using our asymptotic assumptions \eqref{eq:asyg} and \eqref{eq:asyN}, we can evaluate the surface integrals at infinity, obtaining
\begin{align*}
   \lambda\left(2 \sqrt{\lambda m}- \frac{H_{0} \sqrt{\lambda \nu(N)}r_{0}^{2}}{N_0}\right) & \leq 0, \\
   \nu(N)r_{0}^{2}\left(H_{0} N_0+ 4 \nu(N)\right)& \geq 1- N_0^2
\end{align*}
on $\partial M$, where we have used the properties of equipotential photon surfaces and the fact that $\frac{1}{\rho}=\lambda\nu(N)$ on $\partial M$. Recalling $H_{0}>0$ and using $\nu(N)=\frac{m}{r_{0}^{2}}$ by the Smarr formula \eqref{eq:Smarr}, we find
\begin{align}\label{ineq:1}
\lambda\left(\frac{2N_{0}}{H_{0}r_{0}}-1\right)&\leq0,\\\label{ineq:2}
m\left(H_{0}N_{0}+\frac{4m}{r_{0}^{2}}\right)&\geq 1-N_{0}^{2}.
\end{align}
For $\lambda=1$, this gives
\begin{align}\label{ineq:11}
1-N_{0}^{2}&\stackrel{\eqref{ineq:2},\eqref{eq: israel_mc2}}{\leq}  \frac{2m}{r_{0}}\frac{2N_{0}}{H_{0}r_{0}}\stackrel{\eqref{ineq:1}}{\leq} \frac{2m}{r_{0}},\\\label{ineq:12}
1&\stackrel{\eqref{eq: israel_mc2}}{=}\frac{H_{0}r_{0}^{2}}{4N_{0}}\left(H_{0}N_{0}+\frac{4m}{r_{0}^{2}}\right)=\frac{H^{2}_{0}r_{0}^{2}}{4}+\frac{mH_{0}}{N_{0}}\stackrel{\eqref{ineq:1}}{\geq} N_{0}^{2}+\frac{2m}{r_{0}}.
\end{align}
Taken together, this gives $N_{0}^{2}=1-\frac{2m}{r_{0}}$. Similarly, for $\lambda=-1$, \eqref{ineq:11} also follows as $\lambda=\sign(m)$. Instead of arguing as in \eqref{ineq:12}, we use that $N_{0}^{2}>1$ to find
\begin{align*}
m&\stackrel{\eqref{eq: israel_mc2}}{=}\frac{H_{0}r_{0}^{2}}{4N_{0}} m\left(H_{0}N_{0}+\frac{4m}{r_{0}^{2}}\right)\stackrel{\eqref{ineq:2}}{\geq} \frac{H_{0}r_{0}^{2}}{4N_{0}}\left(1-N_{0}^{2}\right)\stackrel{\eqref{ineq:1}}{\geq}\frac{r_{0}}{2}(1-N_{0}^{2}).
\end{align*}
Again, taken together, this gives $N_{0}^{2}=1-\frac{2m}{r_{0}}$. Consequently, in both cases, we have equality in \eqref{ineq:1}, \eqref{ineq:2} and thus in \eqref{IsIn00}, \eqref{IsIn01} from which we learn that all level sets $\Sigma_{N}$ are totally umbilic, that is $\mathring{h}=0$, and have constant $\rho$. Thus, by \eqref{Is3}, we know that the mean curvature $H$ of each level set $\Sigma_{N}$ is constant and given by $H=\frac{\rho'}{\lambda\rho^{2}}$, where $'$ denotes the total $N$-derivative. Thus, by \eqref{Is2}, $\Scal_{\Sigma_{N}}$ is also constant and indeed positive on each level set $\Sigma_{N}$ by the Gauss--Bonnet theorem as $\Sigma_{N}$ is a topological sphere. From this and the definition of the area radius $r(N)$ in \eqref{def:arearadius}, we find (up to a global diffeomorphism on $\mathbb{S}^{2}$) that $\sigma=r(N)^{2}\, g_{\sphere}$ on $\Sigma_{N}$. From the Smarr formula \eqref{eq:Smarr}, we learn that $\rho(N)=\frac{r(N)^{2}}{m}$ on $\Sigma_{N}$. Moreover, suppose towards a contradiction that $r'(N_{1})=0$ for some $N_{1}$ in the image of $N$, i.e., in $[N_{0},1)$ or $(1,N_{0}]$, respectively. Then $\rho'(N_{1})=H(N_{1})=0$ and thus $H'(N_{1})=0$ by \eqref{Is1} which contradicts \eqref{Is2} as $\Sigma_{N_{1}}$ has constant positive Gauss curvature. Hence $r=r(N)$ is invertible and we can write $N=N(r)$ on $[r_{0},\infty)$, giving
\begin{equation*}
g=\frac{r^{4}}{m^{2}}\dot{N}(r)^{2}dr^{2}+r^{2}g_{\sphere}
\end{equation*}
on $[r_{0},\infty)\times\mathbb{S}^{2}$, with $\dot{}$ denoting an $r$-derivative,  and thus in particular spherical symmetry. The claim then follows from Birkhoff's theorem or via explicit computations as those performed in \cite{Ced}.
\end{proof}

\begin{rmk}[Higher dimensions]\label{rmk:high}
As we have applied the Gauss--Bonnet theorem to each level set $\Sigma_{N}$, this proof does not naturally lend itself to an extension to higher dimensions. If one tries to use the same strategy in higher dimensions, one needs to at least assume in addition that all $\Sigma_{N}$ with their induced metrics $\sigma(N)$ are conformal to each other and hence by the prescribed asymptotics conformal to the round sphere. As the condition $\nabla N\neq0$ on $M$ is already rather strong, we do not follow up on this idea, here.
\end{rmk}

\subsection{The zero mass case \`a la Israel}\label{subsec:zeroIsrael}
As discussed in \Cref{sec:zero}, we will exploit the electrostatic potential of $(M,g)$ and work on its level sets, otherwise adapting the strategy of \Cref{subsec:nonzero}. This will lead to the following result.

\begin{theorem}[Degenerate case \`a la Israel]\label{thm:Israeldeg}
Let $(M^{3},g,N)$ be an asymptotically flat static vacuum system of mass $m\in\R$ and decay rate $\tau\geq0$ with respect to asymptotic coordinates $(x^{i})$. Suppose that $M$ has a connected inner boundary $\partial M$ and denote its electrostatic potential by $u$. Suppose further that $\partial M$ arises as a time-slice of a degenerate equipotential photon surface. Assume that $\nabla u\neq0$ on $M$ and that \eqref{eq:asyu} holds asymptotically on $(M,g)$ with respect to $(x^{i})$. Then $(M,g)$ is isometric to the exterior region of a round ball in Euclidean space $(\R^{3},\delta)$, $N\equiv1$ on $M$, and $m=0$.
\end{theorem}

The following corollary is asserted by the proof of \Cref{thm:Israeldeg} and the discussion in \Cref{sec:zero} and gives a new proof of the well-known Willmore inequality in $3$-dimensional Euclidean space \cite[Theorem 1]{Willmore.1968} under the technical assumption that the electrostatic potential $u$ satisfies $\nabla u\neq0$ on $\R^{3}\setminus\overline{\Omega}$. Note however that this technical assumption necessarily implies in particular that $\partial\Omega$ is a topological sphere.

\begin{coroll}[Willmore inequality]\label{thm:Willmore}
Let $\Omega\subset\R^{3}$ be a bounded domain with smooth boundary $\partial\Omega$. Let $u$ be the electrostatic potential of $\Omega$, i.e., the unique smooth solution of $\Delta u=0$ in $\R^{3}\setminus\Omega$ satisfying $u\vert_{\partial\Omega}=1$ and $u\to0$ as $\vert x\vert\to\infty$, where $(x^{i})$ denote the Cartesian coordinates on $\R^{3}$. Assume that $\nabla u\neq0$ on $\R^{3}\setminus\overline{\Omega}$. Then $\partial\Omega$ is diffeomorphic to $\sphere$ and the \emph{Willmore inequality}
\begin{align}\label{eq:Will}
\int_{\partial\Omega} H^{2}\,dA\geq16\pi
\end{align}
holds with equality if and only if $\partial\Omega$ is a round ball. 
\end{coroll}

\begin{rmk}[Extending the Willmore inequality]\label{rmk:Willmore}
In fact, our proofs of \Cref{thm:Israeldeg} and \Cref{thm:Willmore} also show that the Willmore inequality \eqref{eq:Will} holds in any flat, asymptotically flat Riemannian manifold $(M,g)$ with connected boundary $\partial M$ and decay rate $\tau\geq0$, provided that the electrostatic potential $u$ satisfies $\nabla u\neq0$ on $M$ as well as the asymptotic decay condition \eqref{eq:asyu} in the given asymptotically flat coordinates. Moreover, equality holds in \eqref{eq:Will} if and only if $(M,g)$ is isometric to the exterior region of a round ball in $(\R^{3},\delta)$. This is somewhat related to the results in \cite{Ago_Fog_Maz-1}, where the technical assumption $\nabla u\neq0$ on $M$ is not needed.
\end{rmk}

\begin{rmk}[Higher dimensions]
\Cref{rmk:high} applies similarly in this context and illustrates why our proof of the Willmore inequality is specific to $3$ dimensions.
\end{rmk}

\begin{proof}[Proof of \Cref{thm:Israeldeg}]
By the assumption that $\nabla u\neq0$ on $M$, we know that $u$ regularly foliates $M$ into level sets of the same topology. By the Hopf lemma and the maximum principle, we know that $\nabla u\neq0$ on $\partial M$, too. By the asymptotic assumption that \eqref{eq:asyu} holds with respect to $(x^{i})$, we know that these level sets must be topological spheres. In particular, the level set $\partial M=\{u=1\}=\Sigma_{1}$ is a topological sphere. As in the proof of \Cref{thm:Israel}, we use $u$ as a global coordinate and write
\begin{align}\label{eq:gdeg}
g=\rho^{2}du^{2}+\sigma,
\end{align}
where $\sigma$ denotes the Riemannian metric on the level sets $\Sigma_{u}$ and $\rho\definedas\frac{1}{\vert \nabla u\vert}$. From \Cref{lem:zero}, we know that $\Ric=0$. Rewriting $\frac{1}{\rho^{2}}\Ric_{uu}=0$, $\sigma^{IJ}\Ric_{IJ}-\frac{2}{\rho^{2}}\Ric_{uu}=0$, and $\Delta u=0$ with the help of \eqref{eq:gdeg}, we find that
\begin{align}\label{Is1deg}
    0 & = \frac{1}{\rho} \left( H,_u -\frac{\rho}{2} H^2 \right) - \frac{2}{\sqrt{\rho}}\,\Delta_{\sigma} \sqrt{\rho} - \frac{1}{2}\left[ \frac{|\nabla_{\!\sigma} \rho|_{\sigma}^2}{\rho^2} + 2|\mathring{h}|_{\sigma}^2 \right], \\ \label{Is2deg}
    0 & = \frac{H,_u}{\rho} - \Scal_{\Sigma} - \Delta_{\sigma} \log \rho -\left[ \frac{|\nabla_{\!\sigma} \rho|_{\sigma}^2}{\rho^2} + 2|\mathring{h}|_{\sigma}^2 \right], \\ \label{Is3deg}
    0 & = \rho,_u+ \rho^2 H
\end{align}
on $\Sigma_{u}$. Using the same notation as in the proof of \Cref{thm:Israel} and $\partial_{u}\sqrt{\mathfrak{s}}=-\rho H\sqrt{\mathfrak{s}}$, this gives
\begin{align}\label{IsIn00deg}
\partial_u \left( \frac{H \sqrt{\mathfrak{s}}}{\sqrt{\rho}} \right) & \geq 2 \Delta_{\sigma} \sqrt{\rho} \,\sqrt{\mathfrak{s}},
\\ \label{IsIn01deg}
\partial_u \left( \frac{H\sqrt{\mathfrak{s}}}{\rho}\right) & \geq \left(\Scal_{\Sigma} +\Delta_{\sigma} \log \rho\right)\sqrt{\mathfrak{s}} ,
\end{align}
on each level set $\Sigma_{u}$, with equality if and only if the terms in the square brackets in \eqref{Is1deg}, \eqref{Is2deg} vanish on $\Sigma_{u}$. Integrating this as in the proof of \Cref{thm:Israel}, we obtain
\begin{align*}
\int_{\partial M} \frac{H}{\sqrt{\rho}}\,dA-\lim_{r\to\infty}\int_{\sphere_{r}}\frac{H}{\sqrt{\rho}}\,dA&\geq2\int_{0}^{1} \int_{\Sigma_{u}}\Delta_{\sigma} \sqrt{\rho} \,dA\, du=0,\\
\int_{\partial M}\frac{H}{\rho}\,dA-\lim_{r\to\infty}\int_{\sphere_{r}}\frac{H}{\rho}\,dA&\geq\int_{0}^{1}\int_{\Sigma_{u}}\left( \Scal_{\Sigma} +\Delta_{\sigma} \log \rho\right)dA\,du=8\pi
\end{align*}
Exploiting the asymptotic assumptions \eqref{eq:asyg}, \eqref{eq:asyu} as well as \eqref{eq:Smarrb} and the definition of $\rho$, this gives
\begin{align}\label{eq:Will1}
\int_{\partial M} \frac{H}{\sqrt{\rho}}\,dA&\geq 8\pi\sqrt{R_{0}},\\\label{eq:Will2}
\int_{\partial M}\frac{H}{\rho}\,dA&\geq8\pi.
\end{align}
Next, we use that $H_{0}= H=\frac{2}{r_{0}}$ holds on $\partial M$, with $\partial M$ isometric to the round sphere $\sphere_{r_{0}}$ by our considerations in \Cref{sec:zero}. Inserting this into \eqref{eq:Will1} and using \eqref{eq:Smarrb}, we get equality in \eqref{eq:Will1}. By \eqref{Is1deg} or \eqref{Is2deg}, this allows us to conclude that all level sets $\Sigma_{u}$ are totally umbilic, that is $\mathring{h}=0$, and have constant $\rho$. This then implies equality in \eqref{eq:Will2} or in other words $R_{0}=r_{0}$ as expected\footnote{as one would have been immediately able to deduce if we already knew that $(M,g)$ was globally isometric to $(\R^{3}\setminus\overline{B_{r_{0}}},\delta)$ by verifying that $u=\frac{r_{0}}{r}$ is the electrostatic potential of $\Omega=B_{r_{0}}$.}. Moreover, by \eqref{Is3deg}, 
we know that the mean curvature $H$ of each level set $\Sigma_{u}$ is constant and given by $H=-\frac{\rho'}{\rho^{2}}$, where $'$ denotes the total $u$-derivative. Thus, by \eqref{Is2deg}, $\Scal_{\Sigma}$ is also constant and indeed positive on each level set $\Sigma_{u}$ by the Gauss--Bonnet theorem as $\Sigma_{u}$ is a topological sphere. Introducing the area radius
\begin{align*}
r(u) \definedas  \sqrt{\frac{|\Sigma_{u}|}{|\mathbb{S}^2|_{ g_{\sphere}}}},
\end{align*}
 we find (up to a global diffeomorphism on $\mathbb{S}^{2}$) that $\sigma=r(u)^{2}\, g_{\sphere}$ on $\Sigma_{u}$, with $r(1)=r_{0}$. From  \eqref{eq:Smarrb} and $R_{0}=r_{0}$, we learn that $\rho(u)=\frac{r(u)^{2}}{r_{0}}$ on $\Sigma_{u}$. Now suppose towards a contradiction that $r'(u_{1})=0$ for some $u_{1}$ in $(0,1]$. Then $\rho'(u_{1})=H(u_{1})=0$ and thus $H'(u_{1})=0$ by \eqref{Is1deg} which contradicts \eqref{Is2deg} as $\Sigma_{u_{1}}$ has constant positive Gauss curvature. Hence $r=r(u)$ is invertible and we can write $u=u(r)$ on $[r_{0},\infty)$, giving
\begin{equation*}
g=\frac{r^{4}}{r_{0}^{2}}\dot{u}(r)^{2}dr^{2}+r^{2} g_{\sphere}
\end{equation*}
on $[r_{0},\infty)\times\mathbb{S}^{2}$, with $\dot{}$ denoting an $r$-derivative, and thus in particular spherical symmetry. Finally, we know from \eqref{Is1deg}, \eqref{Is2deg}, and $\Scal_{\Sigma_{u(r)}}=\frac{2}{r^{2}}$ (which follows by the Gauss--Bonnet theorem) as well as $H_{0}>0$ and continuity of $H=H(u)$ that
\begin{equation*}
H(u(r))=\frac{2}{r}
\end{equation*}
for all $r\in[r_{0},\infty)$. By the above, this allows us to compute that
\begin{equation*}
\dot{u}(r)=\frac{2r}{r_{0}\rho'(u(r))}\stackrel{\eqref{Is3deg}}{=}-\frac{2r_{0}}{r^{3}H(u(r))}=-\frac{r_{0}}{r^{2}}
\end{equation*}
so that $r^{4}\dot{u}(r)^{2}=r_{0}^{2}$ which implies that $g=dr^{2}+r^{2} g_{\sphere}$ on $[r_{0},\infty)\times\sphere$ (up to an $r$-independent diffeomorphism on $\sphere$). All remaining claims follow from \Cref{lem:zero}.
\end{proof}

\begin{proof}[Proof of \Cref{thm:Willmore}]
Following the proof of \Cref{thm:Israeldeg} for $(M,g)=(\R^{3}\setminus\overline{\Omega},\delta)$, we get that \eqref{eq:Will1} and \eqref{eq:Will2} hold with equality if and only if the terms in the square brackets in \eqref{Is1deg} vanish on all level sets $\Sigma_{u}$. Here, we have exploited that the asymptotic decay assumption \eqref{eq:asyu} holds in this case, see \Cref{sec:zero}. Applying the Cauchy--Schwarz inequality to \eqref{eq:Will1} as well as the identity
\begin{align*}
\int_{\partial M}\vert\nu(u)\vert\,dA=4\pi R_{0}
\end{align*}
which follows from \eqref{eq:Smarrb} and $\nabla u\neq0$ on $\partial M$ as well as the facts that $R_{0}>0$ by the above and $\nu(u)<0$ by the Hopf lemma, we find the Willmore inequality \eqref{eq:Will}. Moreover, rigidity in the Willmore inequality \eqref{eq:Will} holds if and only if rigidity holds in \eqref{eq:Will1} and simultaneously in the Cauchy--Schwarz inequality for $H$ and $\vert\nu(u)\vert$ on $\partial M$. Assume first that equality holds in \eqref{eq:Will} and thus in \eqref{eq:Will1}. Then, arguing as in the proof of \Cref{thm:Israeldeg}, we learn that all level sets $\Sigma_{u}$ are totally umbilic, have constant $\rho$, constant mean curvature $H$, and that we have equality in \eqref{eq:Will2} which gives that all level sets $\Sigma_{u}$ have constant scalar curvature $\Scal_{\Sigma_{u}}$. As this applies in particular to $\partial M$, we have proved that $\partial M=\partial\Omega$ is a round sphere and hence $\Omega$ is a round ball as claimed. Conversely, if $\Omega$ is a round ball, one can check by hand that equality holds in the Willmore inequality \eqref{eq:Will}.
\end{proof}

\section{Photon Surface Uniqueness \`a la Robinson}\label{sec:robinson}
\subsection{The positive and negative mass cases  \`a la Robinson}\label{subsec:nonzeroRob}
In \cite{Rob}, Robinson gave a proof of static vacuum black hole uniqueness for connected horizons. It is a conceptually simple proof based on the divergence theorem and relies on straightforward computations and an ingeniously constructed divergence identity on $(M,g,N)$. We will now generalize this procedure to the non-degenerate equipotential photon surface case and to our much weaker decay assumptions \eqref{eq:asyg}, \eqref{eq:asyN} for decay rate $\tau\geq0$, see \Cref{thm:Robinson} below. Note however that we need to assume that $c>-\frac{1}{2}$ or in other words that the photon surface time-slice $\partial M$ is a topological sphere in the negative mass case.

\begin{theorem}[Non-degenerate case \`a la Robinson]\label{thm:Robinson}
Let $(M^{3},g,N)$ be an asymptotically flat static vacuum system of mass $m\in\R$ and decay rate $\tau\geq0$ with connected inner boundary $\partial M$. Suppose further that $\partial M$ arises as a time-slice of a non-degenerate equipotential photon surface with $c>-\frac{1}{2}$. Then $(M,g)$ is isometric to the exterior region of a round annulus in the Schwarzschild system $(M_{m}^{3},g_{m})$ of mass $m$, with $N$ corresponding to the according restriction of the Schwarzschild lapse $N_{m}$ under this isometry. Furthermore $m\neq0$ and $N_{0}<1$ if $m>0$ while $N_{0}>1$ if $m<0$.
\end{theorem}

Before we present the proof of \Cref{thm:Robinson}, let us first cite some facts due to and from \cite{Rob}, extending them to negative $m$ (i.e., to $\lambda=-1$) which does not change the computations in any essential way.
\begin{lemma}[The Cotton tensor in a static vacuum system]\label{lem:Cotton}
Let $(M^{3},g,N)$ be a static vacuum system and let $C\in\Gamma(T^{*}M^{3})$ be the \emph{Cotton tensor of $(M^{3},g)$} given by
\begin{align}\label{eq:Cotton}
        C_{ijk} :&= \nabla_k\! \Ric_{ij} - \nabla_j \!\Ric_{ik} + \tfrac{1}{4} \left( g_{ik} \nabla_j\!\Scal - g_{ij} \nabla_k\!\Scal \right)
    \end{align}
  in a local frame $\{X_{i}\}$, $i,j,k=1,2,3$. 
 Then, abbreviating
 \begin{align}
 W \definedas \abs{\nabla N}^2,
 \end{align}
 one has
  \begin{align}
  \begin{split}\label{eq:Cottonid}
        \abs{C}^2 &= \frac{8}{N^4}\left( W \abs{\nabla^2 N}^2 - \frac{3}{8} \abs{\nabla W}^2 \right)\\
        &= \frac{4}{N^4} \left( W \Delta W -  \frac{W \nabla^i W \nabla_i N}{N} - \frac{3}{4} \abs{\nabla W}^2 \right).
   \end{split}
   \end{align} 
\end{lemma}

\begin{prop}[Robinson's identity]\label{PropRobId}
Let $(M^{3},g,N)$ be a static vacuum system with $N(p)\neq1$ for all $p\in M$, let $a,b\in\R$, and set $p,q\colon(0,1)\cup(1,\infty)\to\R$ to be the functions given by
\begin{align*}
       p(x) & \definedas  \frac{ax^2 + b}{(1-x^2)^3}, \\
        q(x) & \definedas  -\frac{2a}{(1-x^2)^3} + \frac{6p(x)}{1-x^2}
 \end{align*}
for $x\in(0,1)\cup(1,\infty)$. Then the \emph{Robinson identity}
\begin{equation}\label{RobId}
\diver\!\left(\frac{p(N)}{N} \nabla W + q(N) W \nabla N \right) = \frac{3p(N)}{4NW} \left\vert\nabla W + \frac{8 N W}{1-N^2} \nabla N \right\vert^2 + \frac{p(N) N^3}{4W} |C|^2
    \end{equation}
holds on $M\setminus\CritN$, where $\CritN$ denotes the set of critical points of $N$.
\end{prop}

We will make use of the following corollary which differs slightly from the approach taken by Robinson as we include the negative mass case.

\begin{coroll}[Divergence inequality]\label{coro:divineq}
Let $(M^{3},g,N)$ be a static vacuum system with $N(p)\neq1$ for all $p\in M$, and let $a,b\in\R$ be such that $\lambda p>0$ on $\operatorname{im}(N)$. Then 
\begin{equation}\label{divineq}
\lambda\diver\!\left(\frac{p(N)}{N} \nabla W + q(N) W \nabla N \right) \geq0
 \end{equation}
holds on $M$, with equality if and only if $C\equiv0$ (i.e., $(M,g)$ is locally conformally flat) and 
\begin{equation}\label{eq:prop}
\nabla W=-\frac{8NW}{1-N^{2}}\nabla N
\end{equation}
holds on $M$.
\end{coroll}
\begin{proof}[Proof of \Cref{coro:divineq}]
Clearly, if $\sign(p)=\lambda$ on $\operatorname{im}(N)$, \eqref{divineq} holds on $M\setminus\CritN$ by \eqref{RobId}, with the equality case characterized as claimed. To see that this continues to hold on $\CritN$, we note that the vector field inside the divergence in \eqref{RobId} is smooth on $M$ and that the right hand side of \eqref{RobId} is continuous across critical points of $N$ because $W=0$ on $\CritN$ and hence $W$ attains a global minimum there so that $\nabla W=0$ on $\CritN$. Thus, the Robinson identity \eqref{RobId} and the rigidity characterization extend to $\CritN$.
\end{proof}

We will now proceed to proving \Cref{thm:Robinson}.

\begin{proof}[Proof of \Cref{thm:Robinson}]
First, by non-degeneracy of $\partial M$, we know that $W>0$ on $\partial M$ and $\partial M\cap\CritN=\emptyset$. 
We set $N_{0}\definedas N\vert_{\partial M}$ as before. Then by \Cref{prop:signm}, the non-degeneracy assumption, and the maximum principle, we know that $N_{0}\neq1$ and $N\neq1$ on $M$ so that \Cref{coro:divineq} applies. In particular, we know that $m\neq0$. Now set $W_{0}\definedas W\vert_{\partial M}$ and let $H_{0}$ and $r_{0}$ denote the mean curvature and area radius of $\partial M$ as before. Recall that
\begin{equation}\label{eq:H0id}
H_{0}=-\left.\frac{\nabla^{2}N(\nu,\nu)}{\nu(N)^{3}}\right\vert_{\partial M}=-\lambda\frac{\nabla^{2}N(\nabla N,\nabla N)\vert_{\partial M}}{W_{0}^{\frac{3}{2}}},
\end{equation}
by \eqref{eq:SVE2} and a standard $2+1$-decomposition, using that $\lambda=\sign(\nu(N))=\sign(m)\neq0$ by the above. Integrating the left hand side of \eqref{RobId} over $M$, applying the divergence theorem, and exploiting the asymptotic assumptions \eqref{eq:asyg}, \eqref{eq:asyN} as well as \eqref{eq:H0id}, we find the \emph{$pq$-identity}
\begin{equation}\label{eq: pqid}
\int_{M} \!\diver\! \left( \frac{p(N)}{N} \nabla W + q(N) W \nabla N \right)\!d\mu = 4\pi\!\left(\frac{2H_{0} p_{0} W_{0}r_{0}^{2}}{N_{0}} \!- \!\lambda q_{0} W_{0}^\frac{3}{2} r_{0}^{2} - \frac{a+b}{8m}\right)\!, 
 \end{equation}
where $p_{0}\definedas p(N_{0})$ and $q_{0}\definedas q(N_{0})$, and $d\mu$ denotes the volume element induced by $g$ on $M$. Choosing $a,b$ such that $\lambda(a+b)\geq0$ and $\lambda(aN_{0}^{2}+b)\geq0$ gives $\lambda p>0$ on $\operatorname{im}(N)$ and thus by \Cref{coro:divineq} and \eqref{eq: pqid}, we find
\begin{equation*}
\frac{2\lambda H_{0} p_{0} W_{0}r_{0}^{2}}{N_{0}} - q_{0} W_{0}^\frac{3}{2} r_{0}^{2} \geq \frac{\lambda(a + b)}{8m},
\end{equation*} 
with equality if and only if $C\equiv0$ and \eqref{eq:prop} hold on $M$. Now, we pick $a=1$, $b=-N_{0}^{2}$ which are admissible values of $a,b$ as $\lambda=\sign(1-N_{0}^{2})$. This gives $p_{0}=0$, $q_{0}=-\frac{2}{(1-N_{0}^{2})^{3}}$ and thus, using that $W_{0}=\nu(N)^{2}=\frac{m^{2}}{r_{0}^{4}}$ by the Smarr formula \eqref{eq:Smarr}, we have
\begin{align}\label{eq:R1}
1-N_{0}^{2}\leq \frac{2m}{r_{0}}.
\end{align}
On the other hand, picking $a=-1$, $b=1$ which are also admissible values of $a,b$ as $\lambda=\sign(1-N_{0}^{2})$, we find $p_{0}=\frac{1}{(1-N_{0}^{2})^{2}}$, $q_{0}=\frac{8}{(1-N_{0}^{2})^{3}}$ and thus
\begin{equation}\label{eq:R2}
1-N_{0}^{2}\geq\frac{4mN_{0}}{r_{0}^{2}H_{0}}.
\end{equation}
For $\lambda=1$, \eqref{eq:R1} and \eqref{eq:R2} combine to  
\begin{align*}
\frac{2N_{0}}{r_{0}H_{0}}\leq1
\end{align*}
which allows to conclude that equality holds in both \eqref{eq:R1} and \eqref{eq:R2} as in the proof of \Cref{thm:Israel}, exploiting the photon surface constraints and the fact that $c>0$ by $\sign(c)=\lambda$ so that $\partial M$ is necessarily a round sphere, or in other words applying \eqref{eq: israel_mc2}. For $\lambda=-1$,  \eqref{eq:R1} and \eqref{eq:R2} combine to  
\begin{align}\label{eq:R3}
\frac{2N_{0}}{r_{0}H_{0}}\geq1
\end{align}
and thus
\begin{equation}
m\left(H_{0}N_{0}+\frac{4m}{r_{0}^{2}}\right)\stackrel{\eqref{eq:R3}, \eqref{eq:R1}}{\geq} \frac{2m}{r_{0}}\left(N_{0}^{2}+\frac{2m}{r_{0}}\right)\stackrel{\eqref{eq:R1}}{\geq}1-N_{0}^{2}
\end{equation}
so that we can again conclude  that equality holds in both \eqref{eq:R1} and \eqref{eq:R2} as in the proof of \Cref{thm:Israel}, as we have assumed that $c>-\frac{1}{2}$ so that $\partial M$ has spherical topology and hence \eqref{eq: israel_mc2} holds. This shows equality in \eqref{eq:R1} and \eqref{eq:R2} and thus equality in \eqref{divineq}. Hence by \Cref{coro:divineq}, we have $C\equiv0$ as well as \eqref{eq:prop} on $(M,g)$. From \eqref{eq:prop} and the fact that $W\geq0$, we find that on each connected component $U$ of $\{W\neq0\}$ there is a constant $\alpha>0$ such that
\begin{equation}\label{eq:W(N)}
W=\alpha(1-N^{2})^{4}
\end{equation}
holds on $U$. Next, as $N\neq1$ on $M$ by the above and as $W$ does not vanish entirely e.g.\ because $W_{0}\neq0$, we find that $W\neq0$ on $M$. Using the boundary conditions $W_{0}=\frac{m^{2}}{r_{0}^{4}}$, $1-N_{0}^{2}=\frac{2m}{r_{0}}$, we conclude that
\begin{align*}
W=W(N)&=\frac{\left(1-N^{2}\right)^{4}}{16m^{2}}
\end{align*}
holds on $M$. In particular, there are no critical points of $N$ in $M$. Hence we can introduce coordinates as in the proof of \Cref{thm:Israel} and deduce that \eqref{Is1}, \eqref{Is2}, and \eqref{Is3} hold, with $\rho=W^{-\frac{1}{2}}$. From \eqref{Is3} we find
\begin{equation*}
H=H(N)=\frac{4\lambda N W^{\frac{1}{2}}}{1-N^{2}}>0
\end{equation*} 
on each level set $\Sigma_{N}$ of $N$. Inserting this into \eqref{Is1} gives that each level set $\Sigma_{N}$ is umbilic, while inserting it into \eqref{Is2} gives that $\Sigma_{N}$ has constant scalar curvature $R_{\Sigma_{N}}>0$. One can now conclude as in the proof of \Cref{thm:Israel}.
\end{proof}

\begin{rmk}[Alternative ways of conclusion]
In the above proof, we exploited equality in \eqref{divineq} by deducing the functional dependence \eqref{eq:W(N)} of $W$ on $N$ and then inserting this into Israel's equations \eqref{Is1}--\eqref{Is3}. We have not made use of the vanishing of the Cotton tensor at all. Instead, one can argue as in K\"unzle's work \cite{Kuenzle} and deduce isometry to Schwarzschild from the vanishing of the Cotton tensor. See also \cite{Rob,Hagen,Boucher1984AUT,Lindblom} for related approaches and results.
\end{rmk}

\subsection{The zero mass case  \`a la Robinson}\label{subsec:zeroRobinson}
As in \Cref{subsec:zeroIsrael}, we will exploit the electrostatic potential of $(M,g)$, otherwise adapting the strategy of \Cref{subsec:nonzeroRob}. Most of the work has already been done by Cederbaum and Miehe \cite{Anabel} to prove the Willmore inequality in $(\R^{n},\delta)$ (see \cite[Corollary 1.3]{Anabel}) and we will explain the necessary adaptations, not repeating the technicalities.  The proof in \cite{Anabel} is in $n\geq3$ dimensions and contains a parameter $\beta\geq\frac{n-2}{n-1}$. We would like to point out that our $n=3$ argument works effectively the same way with any $\beta\geq\frac{1}{2}$ and that the proof of the Willmore inequality in \cite[Corollary 1.3]{Anabel} uses $\beta=1$. However, as the analysis is much simpler in the case $\beta=2$, we will only present that case, here\footnote{Other than analytic subtleties related to $\Crit$ which can all be handled exactly the same way as in \cite{Anabel}, the only new ingredient for using $\beta\neq2$ in our proof is the refined Kato inequality (see e.g.\ \cite[Proposition 2.8]{Anabel} and the references given there); the refined Kato inequality applies locally in any Riemannian manifold and hence also in our context.}. Moreover, note that we give a different rigidity argument as the rigidity argument given in the proof of \cite[Theorem 1.7]{Anabel} exploits that $M$ embeds isometrically into $(\R^{3},\delta)$ which we don't know a priori in our setting. This will lead to the following result.

\begin{theorem}[Degenerate case \`a la Robinson]\label{thm:Robinsondeg}
Let $(M^{3},g,N)$ be an asymptotically flat static vacuum system of mass $m\in\R$ and decay rate $\tau\geq0$ with respect to asymptotic coordinates $(x^{i})$. Suppose that $M$ has a connected inner boundary $\partial M$ and denote its electrostatic potential by $u$. Suppose further that $\partial M$ arises as a time-slice of a degenerate equipotential photon surface. Assume  that \eqref{eq:asyu} holds asymptotically on $(M,g)$ with respect to $(x^{i})$. Then $(M,g)$ is isometric to the exterior region of a round ball in Euclidean space $(\R^{3},\delta)$, $N\equiv1$ on $M$, and $m=0$.
\end{theorem}

\begin{rmk}[Willmore inequality]
Our proof of \Cref{thm:Robinsondeg} --- or rather the steps taken to address that $(M,g)$ is not a priori isometric to $(\R^{3}\setminus\overline{\Omega},\delta)$ --- shows that the assertion in \Cref{rmk:Willmore} also holds without the technical assumption $\nabla u\neq0$ on $M$ (but it needs to be transferred to the parameter $\beta=1$ as we are using $\beta=2$, here).
\end{rmk}

\begin{proof}[Proof of \Cref{thm:Robinsondeg}]
First, recall from \Cref{lem:zero} that $\Ric=0$ on $M$ so that $M$ is flat. Moreover, we know that $N\equiv1$ on $M$ and that $m=0$.

In \cite[Theorem 1.7]{Anabel}, Cederbaum and Miehe prove in particular that, for any bounded domain $\Omega\subset\R^{3}$ with smooth connected boundary $\partial\Omega$ and its (Euclidean) electrostatic potential $u\colon\R^{3}\setminus\Omega\to\R$, the divergence inequality 
\begin{align}\label{RobDivId5}
\begin{split}
&\diver\!\left(p(u)\nabla|\nabla u|^{2}+q(u)|\nabla u|^{2}\, \nabla u\right) \\\
&\quad\geq  \frac{3p(u)}{4} \left|\nabla|\nabla u|^2 - \frac{4|\nabla u|^2\,\nabla u}{u}\right|^2
\end{split}
\end{align}
holds on $\R^3\setminus\overline{\Omega}$, with $p,q\colon(0,1]\to\R$ given by
\begin{align}\label{solution F}
p(x)&\definedas (ax+b) x^{-3},\\ \label{solution G}
 q(x)&\definedas-4p(x)x^{-1}+bx^{-4}
 \end{align}
 for any fixed constants $a,b\in\R$ satisfying $a+b\geq0$, $b\geq0$ and all $x\in(0,1]$. Here, $\diver$ denotes the Euclidean divergence. Moreover,
 \begin{align}\label{divPos}
\diver\!\left(p(u)\nabla |\nabla u|^{2}+q(u)|\nabla u|^{2}\, \nabla u\right)\geq 0
 \end{align}
holds on $\R^3\setminus\overline{\Omega}$. In fact, the proof in \cite{Anabel} directly carries over to the setting of \Cref{thm:Robinsondeg}, asserting that \eqref{RobDivId5} and \eqref{divPos} hold on $M$ with respect to the metric $g$ and $p,q$ given by \eqref{solution F}, \eqref{solution G}, respectively. To see this, let us quickly run through the argument given in \cite{Anabel}. To see that \eqref{RobDivId5} and \eqref{divPos} hold, they use harmonicity of $u$ with respect to $\delta$ and the Bochner formula as well as some straightforward local computations relying only on the fact that $(\R^{3}\setminus\overline{\Omega},\delta)$ is flat as well as on \cite[Lemma 2.6]{Anabel} providing ODEs solved by $p$ and $q$. These insights hence carry over directly to our setting and prove the inequalities \eqref{RobDivId5} and \eqref{divPos} on $M$ with respect to $g$. Next, in \cite[Lemma 4.1]{Anabel}, the authors establish that the vector field 
\begin{equation*}
Z\definedas p(u)\nabla|\nabla u|^{2}+q(u)|\nabla u|^{2}\, \nabla u
\end{equation*}
and its divergence $\diver Z$ continuously extend to $\partial\Omega$ and that
\begin{enumerate}
\itemsep0em
\item $\diver Z\in L^{1}(\R^{3}\setminus\Omega)$,
\item $\delta(Z,\nu)\in L^{1}(\{u=u_{0}\}; dA)$ for any regular level set $\{u=u_{0}\}$ of $u$, with
\begin{align*}
\lim_{u_{0}\to0+}\int_{\{u=u_{0}\}}\delta(Z,\nu)\,dA&=4\pi R_{0},
\end{align*}
\item and the divergence theorem
\begin{align*}
\int_{U} \diver Z\,d\mu&=\int_{\partial U} \delta(Z,\eta)\,dA
\end{align*}
holds on any bounded domain $U\subseteq \R^{3}\setminus\Omega$ with smooth boundary $\partial U$ satisfying $\partial U\cap\Crit=\emptyset$. Here, $\eta$ denotes the unit normal to $\partial U$ pointing out of $U$, $dA$ denotes the area measure induced on $\partial U$, and $d\mu$ denotes the volume measure induced on $\R^{3}\setminus\Omega$ by $\delta$.
\end{enumerate}
The continuous extension claims readily transfer to our setting as $\partial M$ is a regular level set of $u$ by the Hopf lemma. Next, we note that $u$ has no critical points near infinity by the assumption \eqref{eq:asyu} and the Smarr-like formula \eqref{eq:Smarrb} which gives $R_{0}\neq0$ by the Hopf lemma. Moreover, by our asymptotic assumption \eqref{eq:asyu} in the asymptotically flat coordinates $(x^{i})$, the asymptotic assertions in \cite[Theorem 2.2 and Proposition 2.5]{Anabel} also hold in our setting. In particular, we have
\begin{align*}
Z^{i}&=-\frac{b x^{i}}{R_{0}\vert x\vert^{3}}+o(\vert x\vert^{-2}),\\
\diver Z&=o(\vert x\vert^{-3})
\end{align*}
as $\vert x\vert\to\infty$. Moreover, we have established Claim 2 in our setting, i.e., on $(M,g)$, namely in  \eqref{eq:Smarrb}. Claims 1 and 3 then follow by low regularity versions of the divergence theorem. As the work of Cheeger--Naber--Valtorta~\cite{Cheeger} and Hardt--Hoffmann-Ostenhof--Hoffmann-Ostenhof--Nadirashvili~\cite{Hardt} extends beyond the Euclidean setup, we know that  $\Crit$ is a set of Lebesgue measure zero and hence $Z$, $\diver Z$, and the normal vector field $\nu=-\frac{\nabla u}{\vert \nabla u\vert}$ extend to $M$ as Lebesgue-measurable functions. As all analyticity claims used in the proof of \cite[Lemma 4.1]{Anabel} (namely \cite[Corollary 2.10 and 2.12]{Anabel}) are local and hence extend to our setting, one can hence perform the cut-off and mollification procedure from the proof of \cite[Lemma 4.1]{Anabel} (in the easier case $\beta=2$) to obtain Claims 1 and 3 also in our setting, i.e., on $(M,g)$.
This gives
\begin{equation}\label{eq:divZdiv}
\int_{M}\diver Z\,d\mu= -\frac{4\pi b}{R_{0}} +2(a+b)\int_{\partial M} H_{0} |\nabla u|^{2}\,dA-(4a+3b)\int_{\partial M} |\nabla u|^{3}\,dA
\end{equation}
for all $a,b\in\R$ with $a+b\geq0$, $b\geq0$. Now recall from \Cref{lem:zero} that $H_{0}=\frac{2}{r_{0}}$ for the area radius $r_{0}>0$ of $\partial M$ and note that the Smarr-like formula \eqref{eq:Smarrb} gives 
\begin{equation*}
\vert\nabla u\vert\vert_{\partial M}=\frac{R_{0}}{r_{0}^{2}}.
\end{equation*}
Inserting this into \eqref{eq:divZdiv} and using that $\diver Z\geq0$ holds on $M$ allows us to conclude that
\begin{equation}\label{eq:4piineq}
 b\leq \frac{4(a+b)R_{0}^{3}}{r_{0}^{3}}-\frac{(4a+3b)R_{0}^{3}}{r_{0}^{4}}
\end{equation}
for all $a,b\in\R$ with $a+b\geq0$, $b\geq0$. Choosing $a=1$, $b=0$, this leads to
\begin{equation*}
R_{0}\leq r_{0}.
\end{equation*}
On the other hand, choosing $a=-1$ and $b=1$, we get
\begin{equation*}
R_{0}\geq r_{0}.
\end{equation*}
In combination, this proves that we have equality in \eqref{eq:4piineq} for these choices of $a,b$. Hence $\diver Z=0$ almost everywhere on $M$ for these choices of $a,b$. From \eqref{RobDivId5} and by continuity of all involved quantities, it then follows that 
\begin{equation}\label{eq:uprop}
\nabla\vert\nabla u\vert^{2}=\frac{4\vert\nabla u\vert^{2}\,\nabla u}{u}
\end{equation}
holds on $M\cup\partial M$. Thus, on each connected component $U$ of $(M\cup\partial M)\setminus\Crit$, there is a positive constant $\kappa_{U}>0$ such that $\vert \nabla u\vert^{2}=\kappa_{U}\,u^{4}$. As $u>0$ on $M\cup\partial M$, this implies that $\Crit=\emptyset$ and
\begin{equation*}
\vert \nabla u\vert^{2}=\frac{u^{4}}{R_{0}^{2}}
\end{equation*}
on $M\cup\partial M$, where we have used the asymptotic decay assumption \eqref{eq:asyu} on $u$. Hence we can use $u$ as a coordinate on $M\cup\partial M$ as in the proof of \Cref{thm:Israeldeg} and write
\begin{equation*}
g=\rho^{2}du^{2}+\sigma,
\end{equation*}
where $\sigma$ denotes the Riemannian metric on the level sets $\Sigma_{u}$ and $\rho\definedas\frac{1}{\vert \nabla u\vert}=\frac{R_{0}}{u^{2}}$. From $\Ric=0$, using \eqref{Is3deg}, \eqref{Is1deg}, and \eqref{Is2deg} from the proof of \Cref{thm:Israeldeg}, we find
\begin{align*}
H=\frac{2u}{R_{0}}, \quad \mathring{h}=0,\quad \Scal_{\Sigma_{u}}=\frac{2u^{2}}{R_{0}^{2}}
\end{align*}
for each (necessarily regular) level set $\Sigma_{u}$ of $u$. This allows us to conclude that $(M,g)$ is isometric to the exterior region of a round ball of radius $r_{0}=R_{0}$ in $(\R^{3},\delta)$ as in the proof of \Cref{thm:Israeldeg}. The remaining claims follow from \Cref{lem:zero}.
\end{proof}

\section{Photon surface uniqueness \`a la Agostiniani and Mazzieri}\label{sec:AM}
\subsection{The positive and negative mass cases  \`a la Agostiniani and Mazzieri}\label{subsec:nonzeroAM}
In \cite{Mazz}, Agostiniani and Mazzieri gave a proof of static vacuum black hole uniqueness for connected horizons. Their proof  is based on a cylindrical ansatz, namely on a \emph{conformal metric $\overline{g}$} conformally related to the given metric $g$ on $M$ which at the end is forced to be cylindrical due to the underlying PDEs, the boundary conditions, and the asymptotic behaviour. The fact that the conformal metric $\overline{g}$ is cylindrical implies that the original static metric $g$ is rotationally symmetric. From this, one can conclude isometry to the Schwarzschild metric e.g.\ by Birkhoff's theorem. 

Proving that the metric $\overline{g}$ is cylindrical is achieved by applying the divergence theorem in $(M,\overline{g})$ to a suitably chosen vector field involving a \emph{pseudo-affine function $\varphi$} outside every level set $\overline{\Sigma}_{s}=\{\varphi=s\}$ of $\varphi$; this leads to a monotonicity formula for a suitably chosen function $\Phi\colon\operatorname{Im}(\varphi)\to\R$. Evoking a splitting principle with splitting direction $\overline{\nabla}\varphi$ then gives the desired rigidity assertion. We will give the adaptation to our situation of the rigidity argument from \cite{Mazz} as well as another rigidity argument more similar to the rigidity argument given in the proof of \Cref{thm:Robinson}.

We will present the relevant parts of the proof from \cite{Mazz} (for $n=3$) applied directly to our situation, i.e., explicitly exploiting the properties of equipotential photon surfaces including the constraints \eqref{normalLapse}, \eqref{scalarMeanCurv} as early on as possible. As \cite{Mazz} very carefully studies singular level sets of $\varphi$, we will not go into detail in these aspects here at all. On the other hand, it is assumed in \cite{Mazz} that $\tau>\frac{1}{2}$ as this allows to allude to the concept of ADM mass. We will show here that this assumption on $\tau$ is not necessary; indeed, all arguments work for $\tau\geq0$ as in the other proofs. Finally, the proof given in \cite{Mazz} contains a parameter $p$ on which they impose various lower bounds; we will only study the case when $p=3$ here; this is fully sufficient for our purposes and simplifies the analysis quite a bit.

\begin{rmk}[Decay assumption in \cite{Mazz}]
Our arguments show that the decay assumption $\tau\geq0$ suffices for proving the claims in \cite{Mazz} for $n=3$, $p=3$, except for making the connection to the ADM mass. They also indicate that the same should hold true for $n>3$, $p\geq3$ (and possibly even smaller values of $p$) as well.
\end{rmk}

We will now generalize the procedure presented in \cite{Mazz} to the non-degenerate equipotential photon surface case, obtaining \Cref{thm:AM} below. Just as in \Cref{thm:Robinson}, we will need to assume $c>-\tfrac{1}{2}$ corresponding to spherical topology of the time-slice of the equipotential photon surface in the negative mass case.
 
\begin{theorem}[Non-degenerate case \`a la Agostiniani and Mazzieri]\label{thm:AM}
Let $(M^{3},g,N)$ be an asymptotically flat static vacuum system of mass $m\in\R$ and decay rate $\tau\geq0$ with connected inner boundary $\partial M$. Suppose further that $\partial M$ arises as a time-slice of a non-degenerate equipotential photon surface with $c>-\frac{1}{2}$. Then $(M,g)$ is isometric to the exterior region of a round annulus in the Schwarzschild system $(M_{m}^{3},g_{m})$ of mass $m$, with $N$ corresponding to the according restriction of the Schwarzschild lapse $N_{m}$ under this isometry. Furthermore $m\neq0$ and $N_{0}<1$ if $m>0$ while $N_{0}>1$ if $m<0$.
\end{theorem}

The following lemma follows from straightforward computations, see \cite[Section 3.1]{Mazz} for the case $0<N<1$. The changes necessary for $N>1$ are obvious. 
\begin{deflemma}[Conformal picture]\label{lem:conformal}
Let $(M^{3},g,N)$ be a static vacuum system with $N(p)\neq1$ for all $p\in M$. The \emph{conformal metric $\overline{g}$} given by 
\begin{equation}\label{eq:conf}
\overline{g} \definedas (1-N^2)^{2} \, g
\end{equation}
is a smooth Riemannian metric on $M$ and the \emph{pseudo-affine function $\varphi\colon M\to\R$} given by
\begin{equation}\label{eq:affine}
\varphi\definedas\ln\left(\frac{1+N}{\vert 1-N\vert}\right)=\ln\left(\frac{\lambda(1+N)}{1-N}\right)
\end{equation}
is a smooth function on $M$, where $\lambda=\sign(m)$ as before. Conversely, one has
\begin{equation}\label{eq:tanh}
N=\left(\tanh\left(\tfrac{\varphi}{2}\right)\right)^{\lambda}.
\end{equation}
The conformal metric $\overline{g}$ and the pseudo-affine function $\varphi$ satisfy 
\begin{align}\label{eq:SVEconf1}
\overline{\Ric} &= \coth (\varphi)\,\overline{\nabla}^2\!\varphi + |\overline{\nabla} \varphi|^2_{\overline{g}}\,\, \overline{g}- d\varphi^2 ,\\\label{eq:SVEconf2}
\overline{\Delta}\varphi&=0
\end{align}
on $M$, where barred quantities are meant to be computed via $\overline{g}$.
\end{deflemma}

For convenience of the reader, let us note that if $(M^{3},g,N)$ is a Schwarzschild system of mass $m\neq0$ then $(M,\overline{g})$ is a round cylinder of radius $2\vert m\vert$ and $\overline{\nabla}\varphi=\frac{1}{2\vert m\vert}\partial_{s}$, where $s=2\vert m\vert\varphi+s_{0}$ for suitable $s_{0}>0$ denotes the height coordinate in this cylinder.

The following lemma is a direct consequence of our decay assumptions and presents a weaker decay version of the asymptotic behavior studied\footnote{choosing $n=3$, $p=3$}  in \cite{Mazz}. See \cite[Lemma 2.5]{CCLP} for details of such decay computations relying only on the decay assumption $\tau\geq0$.

\begin{lemma}[Asymptotics in the conformal picture]\label{lem:confasy}
Let $(M^{3},g,N)$ be an asymptotically flat static vacuum system of mass $m\neq0$ and decay rate $\tau\geq0$. Then
\begin{align}
\overline{g}_{ij}&=\frac{4m^{2}}{\vert x\vert^{2}}\delta_{ij}+o_{2}(\vert x\vert^{-2}),\\
\varphi&=\ln\left(\frac{2\vert x\vert}{\vert m\vert}\right)+o_{2}(1),\\
\sinh\varphi&=\frac{\vert x\vert}{\vert m\vert}(1+o_{2}(1)),\\
\varphi_{,i}&=\frac{x_{i}}{\vert x\vert^{2}}+o_{1}(\vert x\vert^{-1}),\\
\vert\overline{\nabla}\varphi\vert_{\overline{g}}&=\frac{1}{2\vert m\vert}(1+o_{1}(1)),\\
\overline{\nabla}^{2}_{ij}\varphi&=o(\vert x\vert^{-2})
\end{align}
for $i,j=1,2,3$ as $\vert x\vert\to\infty$ with respect to the asymptotic coordinates $(x^{i})$ in the asymptotic end of $M$. Now consider a regular level set $\overline{\Sigma}_{s}\definedas\{\varphi=s\}$ of $\varphi$ and let $\overline{\nu}$ denote the $\overline{g}$-unit normal to $\overline{\Sigma}_{s}$ pointing to the asymptotically cylindrical end. Then
\begin{align}
\overline{\nu}&=\frac{\overline{\nabla} \varphi}{|\overline{\nabla} \varphi|_{\overline{g}}},\\
\overline{\nu}\,^{i}&=\frac{x^{i}}{2\vert m\vert}+o(\vert x\vert)
\end{align}
as $\vert x\vert\to\infty$. The volume element $d\bar{\mu}$ induced on $M$ by $\overline{g}$ behaves as
\begin{align}
d\bar{\mu}&=\frac{8\vert m\vert}{\vert x\vert^{3}}(1+o(1))\,d\mathcal{L},
\end{align}
as $\vert x\vert\to\infty$, where $d\mathcal{L}$ denotes the volume element induced on the asymptotic end of $M$ by the flat metric $\delta$ in the coordinates $(x^{i})$. Moreover, the area element $d\bar{A}$ induced on the level set $\overline{\Sigma}_{s}$ by $\overline{g}$ and the area $\vert\overline{\Sigma}_{s}\vert_{\bar{\sigma}}$ of $\overline{\Sigma}_{s}$ with respect to $\overline{g}$ satisfy
\begin{align}
d\bar{A}&=\left(1-\left(\tanh\left(\tfrac{s}{2}\right)\right)^{2\lambda}\right)^{2}\,dA=4m^{2}(1+o(1))\,d\Omega,\\\label{eq:surfarea}
\vert\overline{\Sigma}_{s}\vert_{\bar{\sigma}}&=\left(1-\left(\tanh\left(\tfrac{s}{2}\right)\right)^{2\lambda}\right)^{2}\,\vert\Sigma_{\left(\tanh\left(\tfrac{s}{2}\right)\right)^{\lambda}}\vert=16\pi m^{2}(1+o(1))
\end{align}
as $s\to\infty$. Here, $d\Omega$ denotes the canonical area element on $\sphere$ with respect to ambient Cartesian coordinates $(x^{i})$, and where $\Sigma_{\left(\tanh\left(\tfrac{s}{2}\right)\right)^{\lambda}}=\{N=\left(\tanh\left(\tfrac{s}{2}\right)\right)^{\lambda}\}$ and $\vert \Sigma_{\left(\tanh\left(\tfrac{s}{2}\right)\right)^{\lambda}}\vert$ denotes its area with respect to~$\sigma$. 
 \end{lemma}

\begin{proof}[Proof of \Cref{thm:AM}]
Let us consider the smooth vector field
\begin{equation*}
X\definedas  \frac{\overline{\nabla} |\overline{\nabla} \varphi|^{2}_{\overline{g}}}{\sinh \varphi}
\end{equation*} 
on $M$ and let us remark that $(M,\overline{g})$ is geodesically and metrically complete up to $\partial M$. Using \eqref{eq:SVEconf1}, \eqref{eq:SVEconf2} and the Bochner formula, we find that
\begin{align}\label{eq:divbarX}
\overline{\diver}\, X &= \frac{\overline{\Delta} |\overline{\nabla} \varphi|^{2}_{\overline{g}} - \overline{g}(\overline{\nabla} |\overline{\nabla} \varphi|^{2}_{\overline{g}},  \overline{\nabla} \ln (\sinh \varphi))}{\sinh \varphi} = \frac{ 2|\overline{\nabla}^2 \varphi|^2_{\overline{g}}}{\sinh \varphi} \geq 0
\end{align}
on $M$. Exploiting the asymptotic behavior derived in \Cref{lem:confasy}, one finds that $\overline{\diver}\,X$ is integrable on $M$ with respect to $d\bar{\mu}$. Next, we apply the divergence theorem on $M$ with respect to $\overline{g}$ to find
\begin{align*}
\int_{M}\overline{\diver}\,X\,d\bar{\mu}= \lim_{s\to\infty}\int_{\overline{\Sigma}_{s}}\frac{\overline{g}(\overline{\nabla} |\overline{\nabla} \varphi|^{2}_{\overline{g}}, \, \frac{\overline{\nabla} \varphi}{|\overline{\nabla} \varphi|_{\overline{g}}})}{\sinh \varphi} \,  \,d\bar{A}-\int_{\partial M}\frac{\overline{g}(\overline{\nabla} |\overline{\nabla} \varphi|^{2}_{\overline{g}}, \, \frac{\overline{\nabla} \varphi}{|\overline{\nabla} \varphi|_{\overline{g}}})}{\sinh \varphi} \,  \,d\bar{A}.
\end{align*}
Here, we have used that $\partial M$ is a regular level set of $N$ and hence of $\varphi$ and that $\overline{\Sigma}_{s}$ is regular for all $s\geq s_{1}$ and a suitably large $s_{1}$. Appealing again to \Cref{lem:confasy}, we see that the above boundary integral at infinity is well-defined and indeed vanishes. In combination, the above facts combined with $\overline{g}$-harmonicity of $\varphi$ and the properties of time-slices of non-degenerate equipotential photon surfaces give
\begin{align*}
0\leq -\int_{\partial M}\frac{\overline{g}(\overline{\nabla} |\overline{\nabla} \varphi|^{2}_{\overline{g}}, \, \frac{\overline{\nabla} \varphi}{|\overline{\nabla} \varphi|_{\overline{g}}})}{\sinh \varphi} \,d\bar{A}=\frac{2\overline{H}_{0}\vert\overline{\nabla}\varphi\vert_{\overline{g}}^{2}\vert_{\partial M}}{\sinh(\varphi_{0})} \,\vert\partial M\vert_{\overline{\sigma}},
\end{align*}
with $\varphi_{0}$ and $\overline{H}_{0}$ denoting the value of $\varphi$ on $\partial M$ and the mean curvature of $\partial M$ with respect to the normal $\overline{\nu}$ and the metric $\overline{g}$, respectively. This can be summarized as
\begin{equation}\label{eq:signHbar}
\overline{H}_{0}\geq0.
\end{equation}
Moreover, by \eqref{eq:divbarX}, equality holds in \eqref{eq:signHbar} if and only if $\overline{\nabla}^{2}\varphi=0$ on $M$. Repeating the same divergence theorem argument on the domain $\{s>s_{0}\}$ for some regular value $s_{0}$ of $\varphi$, one finds
\begin{equation}\label{eq:nonneg}
\int_{\overline{\Sigma}_{s_{0}}}\overline{g}(\overline{\nabla} |\overline{\nabla} \varphi|^{2}_{\overline{g}}, \, \frac{\overline{\nabla} \varphi}{|\overline{\nabla} \varphi|_{\overline{g}}}) \,d\bar{A}\leq0
\end{equation}
for all regular values $s_{0}$ of $\varphi$. Following the analytic arguments in \cite{Mazz}, this also holds for singular values $s_{0}$ of $\varphi$. We now consider the function $\Phi\colon\operatorname{Im}(\varphi)\to\R$ given by
\begin{equation}\label{eq:Phi}
\Phi(s)\definedas\int_{\overline{\Sigma}_{s}}\vert\overline{\nabla}\varphi\vert_{\overline{g}}^{3}\;d\bar{A}.
\end{equation}
A rather standard computation shows that $\Phi$ is continuous and differentiable at regular values $s$ of $\varphi$, with
\begin{equation}\label{eq:Phi'}
\Phi'(s)\definedas\int_{\overline{\Sigma}_{s}}\overline{g}(\overline{\nabla}\vert\overline{\nabla}\varphi\vert_{\overline{g}}^{2},\frac{\overline{\nabla}\varphi}{\vert\overline{\nabla}\varphi\vert_{\overline{g}}})\;d\bar{A}
\end{equation}
on regular level sets. Following the arguments in \cite{Mazz}, one can see that $\Phi$ is also differentiable at critical values $s$ of $\varphi$ with \eqref{eq:Phi'} extending also to critical values. By \eqref{eq:nonneg}, this shows $\Phi'\leq0$ so that $\Phi$ is monotonically decreasing. Combining the Smarr formula \eqref{eq:Smarr} and the fact that $\overline{\Sigma}_{\varphi_{0}}=\partial M=\Sigma_{N_{0}}$ with the properties of time-slices of equipotential photon surfaces and the asymptotic decay asserted in \Cref{lem:confasy}, we deduce that
\begin{equation*}
\frac{32\pi\vert m\vert \nu(N)^{2}}{(1-N_{0}^{2})^{4}}=\Phi(\varphi_{0})\geq\lim_{s\to\infty}\Phi(s)=\frac{2\pi}{\vert m\vert}.
\end{equation*}
Together with \eqref{eq:signHbar} transformed back into the original variables, we thus have
\begin{align}\label{eq:oneside}
H_{0}&\geq \frac{4N_{0}\,\nu(N)}{1-N_{0}^{2}},\\\label{eq:opposite}
\vert\nu(N)\vert&\geq\frac{(1-N_{0}^{2})^{2}}{4\vert m\vert}
\end{align}
on $\partial M$. Exploiting the photon surface constraint \eqref{normalLapse} and recalling that $c\neq0$ by non-degeneracy of $\partial M$ as well as $\sign(c)=\lambda=\sign(1-N_{0}^{2})$ and\footnote{We would like to point out that $H_0>0$ actually follows from \eqref{eq:oneside} directly in this approach.} $H_{0}>0$, we rewrite \eqref{eq:oneside} as
\begin{equation}\label{eq:first}
\lambda N_{0}\leq \frac{\lambda}{\sqrt{2c+1}}.
\end{equation}
Next, we consider the area radius $r_{0}$ of $\partial M$ so that by the Gauss--Bonnet theorem (recall $c>-\tfrac{1}{2}$), the photon surface constraints \eqref{normalLapse}, \eqref{scalarMeanCurv}, and the Smarr formula \eqref{eq:Smarr}, we have
\begin{equation}\label{eq:Smarrnew}
\frac{m^{2}}{r_{0}^{2}}=\frac{c^{2}N_{0}^{2}}{2c+1}.
\end{equation}
Thus, from \eqref{eq:opposite} and \eqref{eq:first}, we find
\begin{equation}\label{eq:second}
c^{2}N_{0}^{2}\geq\frac{2c+1}{4}(1-N_{0}^{2})^{2}.
\end{equation}
For $\lambda=1$, we have $N_{0}\leq\frac{1}{\sqrt{2c+1}}$ from \eqref{eq:first} and $N_{0}\geq\frac{1}{\sqrt{2c+1}}$ from \eqref{eq:second} with $N_{0}\leq\frac{1}{\sqrt{2c+1}}$ applied to the left hand side. For $\lambda=-1$, \eqref{eq:first} gives $N_{0}\geq\frac{1}{\sqrt{2c+1}}$ and \eqref{eq:second} gives $N_{0}\leq\frac{1}{\sqrt{2c+1}}$ when applying $N_{0}\geq\frac{1}{\sqrt{2c+1}}$ to the right hand side. Thus, we have learned that
\begin{equation}\label{eq:N0}
N_{0}=\frac{1}{\sqrt{2c+1}}
\end{equation}
regardless of the value of $\lambda$. From \eqref{eq:N0} and \eqref{eq:Smarrnew} respectively from the photon surface constraint \eqref{scalarMeanCurv}, we find
\begin{align*}
\frac{m}{r_{0}}&=\frac{c}{2c+1},\\
H_{0}r_{0}&=\frac{2}{\sqrt{2c+1}}.
\end{align*}
Combined with \eqref{eq:N0}, these give equality in \eqref{eq:oneside} which is equivalent to $\overline{H}_{0}=0$. Hence we have equality in \eqref{eq:signHbar} which demonstrates that $\overline{\nabla}^{2}\varphi=0$ on $M$. We will now present two options how to conclude from here, mostly in order to highlight similarities with the proofs of \Cref{thm:Israel,thm:Robinson}: First, similar to the line of thoughts in \cite{Mazz}, we know from $\overline{\nabla}^{2}\varphi=0$ on $M$ that $\overline{\nabla}\varphi$ is a parallel vector field on $(M,\overline{g})$ and thus has constant length $\vert\overline{\nabla}\varphi\vert_{\overline{g}}$. From the asymptotic considerations in \Cref{lem:confasy}, we find that
\begin{equation*}
\vert\overline{\nabla}\varphi\vert_{\overline{g}}=\frac{1}{2\vert m\vert}
\end{equation*}
on $M$. In particular, $\varphi$ has no critical points in $M$ and we can use local coordinates $(\varphi,y^{I})$, $I=1,2$, as in the proof of \Cref{thm:Israel} to write
\begin{equation*}
\overline{g}=4m^{2}d\varphi^{2}+\overline{\sigma}.
\end{equation*}
This leads to
\begin{equation*}
0=\overline{\nabla}^{2}_{IJ}\,\varphi=-\overline{\Gamma}^{\,\varphi}_{IJ}=\frac{\overline{\sigma}_{IJ,\varphi}}{8m^{2}}
\end{equation*}
for $I,J=1,2$ which shows that $\overline{\sigma}$ is independent of $\varphi$. Thus $\overline{\sigma}$ coincides with the metric on $\partial M=\overline{\Sigma}_{\varphi_{0}}$ which is known to be round by the properties of time-slices of equipotential photon surfaces and the Gau{\ss}--Bonnet theorem. This shows that $\overline{g}=4m^{2}d\varphi^{2}+g_{\sphere_{\overline{r}_{0}}}$ up to a global ($\varphi$-independent) diffeomorphism on the spherical factor, where $\overline{r}_{0}$ denotes the area radius of $\partial M$ with respect to $\overline{\sigma}$. This shows that $g$ is spherically symmetric; as the spheres of symmetry of $g$ correspond exactly to the levels sets of $N$, we can conclude by Birkhoff's theorem. More directly, suppressing the global diffeomorphism, we can learn from \eqref{eq:surfarea} that $\overline{r}_{0}=2\vert m\vert$ and conclude that $M\approx(2\vert m\vert,\infty)\times\sphere$, $\partial M\approx \{2\vert m\vert\}\times\sphere$, and $\overline{g}=4m^{2}(d\varphi^{2}+g_{\sphere})$. Setting $r(N)\definedas \frac{2m}{1-N^{2}}$, we find $2\vert m\vert = \overline{r}_{0}=\lambda (1-N_{0}^{2})r_{0}$ so that $r_{0}=r(N_{0})$ as well $N(r)=\sqrt{1-\frac{2m}{r}}$ and thus we find $M=(r_{0},\infty)\times\sphere$  and
\begin{equation*}
g=\frac{4m^{2}}{(1-N^{2})^{2}}\left(\frac{4}{(1-N^{2})^{2}}dN^{2}+g_{\sphere}\right)=\frac{r^{4}}{m^{2}}dN^{2}+r^{2}g_{\sphere}=\frac{1}{N^{2}}dr^{2}+r^{2}g_{\sphere}
\end{equation*}
with $N=N(r)=\sqrt{1-\frac{2m}{r}}$ as desired. 

The second proof for concluding isometry to Schwarzschild from $\overline{\nabla}^{2}\varphi=0$ on $M$ goes as follows: From the definition of $\varphi$ and $\overline{g}$, one computes
\begin{equation*}
0=\overline{\nabla}^{2}\varphi=\frac{2}{1-N^{2}}\nabla^{2}N+\frac{4N}{(1-N^{2})^{2}}(3dN^{2}-\vert\nabla N\vert^{2}g)
\end{equation*}
Plugging in $\nabla N$, this shows that $\nabla\vert\nabla N\vert^{2}$ and $\nabla N$ are parallel, more precisely that they satisfy \eqref{eq:prop} which also appears in the proof of \Cref{thm:Robinson}. We can then conclude as in said proof.
 \end{proof}

\subsection{The zero mass case \`a la Agostiniani and Mazzieri}\label{subsec:zeroAM}
As in \Cref{subsec:zeroIsrael,subsec:zeroRobinson}, we will exploit the electrostatic potential of $(M,g)$, otherwise adapting the strategy of \Cref{subsec:nonzeroAM}. Most of the work has already been done by Agostiniani and Mazzieri \cite{MazzAgo} to prove the Willmore inequality in $(\R^{n},\delta)$ and we will explain the necessary adaptations, not repeating the analytic subtleties. The proof in \cite{MazzAgo} is in $n\geq3$ dimensions and contains a parameter $\beta\geq\frac{n-2}{n-1}$ which plays the same role as the parameter $\beta$ in \cite{Anabel}, see \Cref{subsec:zeroRobinson}. We would like to point out that our argument for $n=3$ works effectively the same way with any $\beta\geq\frac{1}{2}$ and that the proof of the Willmore inequality in \cite{MazzAgo} uses $\beta=1$. However, as the analysis is much simpler in the case $\beta=2$, we will only present that case, here\footnote{Other than analytic subtleties related to $\Crit$ which can all be handled exactly the same way as in \cite{MazzAgo}, the only new ingredient for using $\beta\neq2$ in our proof is again the refined Kato inequality which applies locally in any Riemannian manifold and hence also in our context.}. As in \Cref{subsec:nonzeroAM},  we give a slight adaptation of the rigidity argument given in \cite{MazzAgo}; moreover, as in \Cref{subsec:zeroRobinson}, we also give an alternative rigidity argument showing the similarity to the rigidity argument in \Cref{subsec:nonzeroRob}. This will lead to the following result.

\begin{theorem}[Degenerate case \`a la Agostiniani--Mazzieri]\label{thm:AMdeg}
Let $(M^{3},g,N)$ be an asymptotically flat static vacuum system of mass $m\in\R$ and decay rate $\tau\geq0$ with respect to asymptotic coordinates $(x^{i})$. Suppose that $M$ has a connected inner boundary $\partial M$ and denote its electrostatic potential by $u$. Suppose further that $\partial M$ arises as a time-slice of a degenerate equipotential photon surface. Assume  that \eqref{eq:asyu} holds asymptotically on $(M,g)$ with respect to $(x^{i})$. Then $(M,g)$ is isometric to the exterior region of a round ball in Euclidean space $(\R^{3},\delta)$, $N\equiv1$ on $M$, and $m=0$.
\end{theorem}

\begin{rmk}[Willmore inequality]
Our proof of \Cref{thm:AMdeg} --- or rather the steps taken to address that $(M,g)$ is not a priori isometric to $(\R^{3}\setminus\overline{\Omega},\delta)$ --- suggests that the assertion in \Cref{rmk:Willmore} also holds without the technical assumption $\nabla u\neq0$ on $M$ (but it needs to be transferred to the parameter $\beta=1$ as we are using $\beta=2$, here).
\end{rmk}

\begin{proof}[Proof of \Cref{thm:AMdeg}]
First, recall from \Cref{lem:zero} that $\Ric=0$ on $M$ so that $M$ is flat. Moreover, we already know that $N\equiv 1$ on $M$ and that $m=0$. Next, we note that $\partial M$ is a regular level set of $u$ by the Hopf lemma and that the Smarr-like formula \eqref{eq:Smarrb} gives $R_{0}\neq0$ by the Hopf lemma. Now, as in \cite{MazzAgo}, we perform a conformal change to the \emph{conformal metric}
\begin{equation}\label{eq:gconf}
\overline{g}\definedas u^{2}g
\end{equation}
and define the \emph{pseudo-affine function}
\begin{equation}\label{eq:defphizero}
\varphi\definedas -\ln u
\end{equation}
on $M$. Conversely, one has
\begin{equation}\label{eq:e-}
u=e^{-\varphi}.
\end{equation}
It is straightforward to compute (see \cite[Section 2]{MazzAgo} for details) that
\begin{align}\label{eq:conf1}
\overline{\Ric}-\overline{\nabla}^{2}\varphi+d\varphi^{2} &=\vert\overline{\nabla}\varphi\vert_{\overline{g}}^{2}\;\overline{g},\\\label{eq:conf2}
\overline{\Delta}\varphi&=0
\end{align}
hold on $M$ by the static vacuum equations \eqref{eq:SVE1}, \eqref{eq:SVE2}. From the boundary condition $u\vert_{\partial M}=1$ we learn that
\begin{equation}\label{eq:phi0}
\varphi\vert_{\partial M}=0,
\end{equation}
so that in particular $\partial M$ is a level set of $\varphi$. From $d\varphi =-\frac{du}{u}$, we see that $\partial M$ is indeed a regular level set of $\varphi$ as well. Next, from the asymptotic assumption $u=\frac{R_{0}}{\vert x\vert}(1+o_{2}(1))$ on the electrostatic potential $u$, it follows that
\begin{align}
\overline{g}_{ij}&=\frac{R_{0}^{2}}{\vert x\vert^{2}}\,\delta_{ij}(1+o_{2}(1)),\\
\varphi&=\ln\left(\frac{\vert x\vert}{R_{0}}\right)+o_{2}(1),\\
e^{\varphi}&=\frac{\vert x\vert}{R_{0}}(1+o_{2}(1)),\\
\varphi_{,i}&=\frac{x_{i}}{\vert x\vert^{2}}+o_{1}(\vert x\vert^{-1}),\\
\vert\overline{\nabla}\varphi\vert_{\overline{g}}&=\frac{1}{R_{0}}(1+o_{1}(1)),\\
\overline{\nabla}^{2}_{ij}\varphi&=o(\vert x\vert^{-2})
\end{align}
for $i,j=1,2,3$ as $\vert x\vert\to\infty$. Now consider a regular level set $\overline{\Sigma}_{s}\definedas\{\varphi=s\}$ of $\varphi$ and let $\overline{\nu}$ denote the $\overline{g}$-unit normal to $\overline{\Sigma}_{s}$ pointing to the asymptotically cylindrical end. Then
\begin{align}
\overline{\nu}&=\frac{\overline{\nabla} \varphi}{|\overline{\nabla} \varphi|_{\overline{g}}},\\
\overline{\nu}\,^{i}&=x^{i}(1+o(1))
\end{align}
as $\vert x\vert\to\infty$. The volume element $d\bar{\mu}$ induced on $M$ by $\overline{g}$ behaves as
\begin{align}
d\bar{\mu}&=\frac{R_{0}^{3}}{\vert x\vert^{3}}(1+o(1))\,d\mathcal{L},
\end{align}
as $\vert x\vert\to\infty$ respectively $r\to\infty$, where $d\mathcal{L}$ denotes the volume element induced on the asymptotic end of $M$ induced by the flat metric $\delta$ in the coordinates $(x^{i})$. Moreover, the area element $d\bar{A}$ induced on the level set $\overline{\Sigma}_{s}$ by $\overline{g}$ and the area $\vert\overline{\Sigma}_{s}\vert_{\bar{\sigma}}$ of $\overline{\Sigma}_{s}$ with respect to $\overline{g}$ satisfy
\begin{align}
d\bar{A}&=e^{-2\varphi}dA=R_{0}^{2}(1+o(1))\,d\Omega,\\\label{eq:surfarea0}
\vert\overline{\Sigma}_{s}\vert_{\bar{\sigma}}&=e^{-2\varphi}\,\vert\Sigma_{e^{-\varphi}}\vert=4\pi R_{0}^{2}(1+o(1))
\end{align}
as $s\to\infty$, where $d\Omega$ denotes the canonical area element on $\sphere$ with respect to $(x^{i})$, $\Sigma_{e^{-\varphi}}=\{u=e^{-\varphi}\}$, and $\vert \Sigma_{e^{-\varphi}}\vert$ denotes its area with respect to $\sigma$. Next, let us consider the smooth vector field
\begin{equation}\label{def:Y}
Y\definedas\frac{\overline{\nabla}\vert\overline{\nabla}\varphi\vert_{\overline{g}}^{2}}{e^{\varphi}}
\end{equation}
on $M$ and let us remark that $(M,\overline{g})$ is geodesically and metrically complete up to $\partial M$. Its divergence with respect to $\overline{g}$ satisfies
\begin{equation}\label{eq:divY}
\overline{\diver}\,Y=\frac{2\vert\overline{\nabla}^{2}\varphi\vert_{\overline{g}}}{e^{\varphi}}\geq0
\end{equation}
on $M$, where we have used the Bochner formula as well as \eqref{eq:conf1}, \eqref{eq:conf2}. Exploiting the asymptotic behavior derived above, one finds that $\overline{\diver}\,Y$ is integrable on $M$ with respect to $d\bar{\mu}$. Next, we apply the divergence theorem on $M$ with respect to $\overline{g}$ to find
\begin{align*}
\int_{M}\overline{\diver}\,Y\,d\bar{\mu}= \lim_{s\to\infty}\int_{\overline{\Sigma}_{s}}\frac{\overline{g}(\overline{\nabla} |\overline{\nabla} \varphi|^{2}_{\overline{g}}, \, \frac{\overline{\nu} \varphi}{|\overline{\nabla} \varphi|_{\overline{g}}})}{e^{\varphi}} \,  \,d\bar{A}-\int_{\partial M}\frac{\overline{g}(\overline{\nabla} |\overline{\nabla} \varphi|^{2}_{\overline{g}}, \, \frac{\overline{\nabla} \varphi}{|\overline{\nabla} \varphi|_{\overline{g}}})}{e^{\varphi}} \,  \,d\bar{A},
\end{align*}
where we have used that $\partial M$ is a regular level set of $\varphi$  and that $\overline{\Sigma}_{s}$ is regular for all $s\geq s_{1}$ and a suitably large $s_{1}$. Appealing again to the above decay assertions, we see that the above boundary integral at infinity is well-defined and indeed vanishes. In combination, the above facts combined with $\overline{g}$-harmonicity of $\varphi$, the Smarr-like formula \eqref{eq:Smarrb} leading to $\vert\overline{\nabla}\varphi\vert_{\overline{g}}\vert_{\partial M}=\frac{R_{0}}{r_{0}^{2}}$ and the properties of time-slices of degenerate equipotential photon surfaces established in \Cref{lem:zero} give
\begin{align*}
0&\leq -\int_{\partial M}\frac{\overline{g}(\overline{\nabla} |\overline{\nabla} \varphi|^{2}_{\overline{g}}, \, \frac{\overline{\nabla} \varphi}{|\overline{\nabla} \varphi|_{\overline{g}}})}{e^{\varphi}} \,d\bar{A}=2\int_{\partial M} \overline{H}_{0}\vert\overline{\nabla}\varphi\vert_{\overline{g}}^{2}\,d\bar{A}=\frac{8\pi r_{0}^{4} \overline{H}_{0}}{R_{0}^{4}},
\end{align*}
with $\overline{H}_{0}$ denoting the (necessarily constant) mean curvature of $\partial M$ with respect to $\overline{g}$. This can be summarized as $\overline{H}_{0}\geq0$ and translates to
\begin{equation*}
\frac{2}{r_{0}}=H_{0}\geq2\vert\nabla u\vert=\frac{2R_{0}}{r_{0}^{2}}
\end{equation*}
or, equivalently,
\begin{equation}\label{eq:signH}
r_{0}\geq R_{0}
\end{equation}
by transforming back to the original variables and appealing to \eqref{lem:zero} and the Smarr-like formula \eqref{eq:Smarrb}. Moreover, by \eqref{eq:divY}, equality holds in \eqref{eq:signH} if and only if $\overline{\nabla}^{2}\varphi=0$ on $M$. Repeating the same divergence theorem argument on the domain $\{s>s_{0}\}$ for some regular value $s_{0}$ of $\varphi$, one finds
\begin{equation}\label{eq:nonneg0}
\int_{\overline{\Sigma}_{s_{\tiny0}}}\overline{g}(\overline{\nabla} |\overline{\nabla} \varphi|^{2}_{\overline{g}}, \, \frac{\overline{\nabla} \varphi}{|\overline{\nabla} \varphi|_{\overline{g}}}) \,d\bar{A}\leq0
\end{equation}
for all regular values $s_{0}$ of $\varphi$. Following the analytic arguments in \cite{MazzAgo}, this also holds for singular values $s_{0}$ of $\varphi$ using the work of Cheeger--Naber--Valtorta~\cite{Cheeger} and Hardt--Hoffmann-Ostenhof--Hoffmann-Ostenhof--Nadirashvili~\cite{Hardt} which extends beyond the Euclidean setup. Using the same function $\Phi$ of $\varphi$ given by \eqref{eq:Phi} as in the proof of \Cref{thm:AM}, we again find that $\Phi'(s)\leq0$ for all values $s$ of $\varphi$, i.e., all $s\in(0,\infty)$, so that $\Phi$ is monotonically decreasing. Using the fact that $\overline{\Sigma}_{0}=\partial M=\Sigma_{1}$ and the asymptotic decay asserted above, we deduce that
\begin{equation*}
\frac{4\pi R_{0}^{3}}{r_{0}^{4}}=\Phi(0)\geq\lim_{s\to\infty}\Phi(s)=\frac{4\pi}{R_{0}}
\end{equation*}
or in other words
\begin{equation}\label{eq:signHopp}
R_{0}\geq r_{0}.
\end{equation}
Combining \eqref{eq:signH} with \eqref{eq:signHopp}, we have $R_{0}=r_{0}$ and can conclude that $\overline{\nabla}^{2}\varphi=0$ from equality in \eqref{eq:signH}. We will now present two options how to conclude from here, mostly in order to highlight similarities with the proofs of \Cref{thm:Israel,thm:Robinson}: First, similar to the line of thoughts in \cite{MazzAgo}, we know from $\overline{\nabla}^{2}\varphi=0$ on $M$ that $\overline{\nabla}\varphi$ is a parallel vector field on $(M,\overline{g})$ and thus has constant length $\vert\overline{\nabla}\varphi\vert_{\overline{g}}$. From the above asymptotic considerations in, we find that
\begin{equation*}
\vert\overline{\nabla}\varphi\vert_{\overline{g}}=\frac{1}{R_{0}}
\end{equation*}
on $M$. In particular, $\varphi$ has no critical points in $M$ and we can use local coordinates $(\varphi,y^{I})$, $I=1,2$, as in the proof of \Cref{thm:Israel} to write
\begin{equation*}
\overline{g}=R_{0}^{2}\,d\varphi^{2}+\overline{\sigma}.
\end{equation*}
This leads to
\begin{equation*}
0=\overline{\nabla}^{2}_{IJ}\,\varphi=-\overline{\Gamma}^{\,\varphi}_{IJ}=\frac{\overline{\sigma}_{IJ,\varphi}}{2R_{0}^{2}}
\end{equation*}
for $I,J=1,2$ which shows that $\overline{\sigma}$ is independent of $\varphi$. Thus $\overline{\sigma}$ coincides with the metric $\overline{\sigma}=\sigma$ on $\partial M=\overline{\Sigma}_{\varphi_{0}}$ which is known to be round by the properties of time-slices of degenerate equipotential photon surfaces and the Gau{\ss}--Bonnet theorem. This shows that $\overline{g}=R_{0}^{2}\,d\varphi^{2}+g_{\sphere_{R_{0}}}$ up to a global ($\varphi$-independent) diffeomorphism on the spherical factor. Next, suppressing the global diffeomorphism, we find from the definitions of $\varphi$ and $\overline{g}$ in \eqref{eq:gconf}, \eqref{eq:defphizero} that $g=u^{-2}\,\overline{g}=R_{0}^{2}(u^{-4}du^{2}+u^{-2}g_{\sphere})$. Setting $r(u)\definedas \frac{R_{0}}{u}$ then gives $g=dr^{2}+r^{2}g_{\sphere}$ on $M\approx(R_{0},\infty)\times\sphere$ as claimed.

The second proof for concluding isometry to the exterior region of a ball from $\overline{\nabla}^{2}\varphi=0$ on $M$ goes as follows: From the definition of $\varphi$ and $\overline{g}$ in \eqref{eq:gconf}, \eqref{eq:defphizero}, one computes
\begin{equation*}
0=\overline{\nabla}^{2}\varphi=-\frac{\nabla^2u}{u}+\frac{1}{u^2}\left(3du^2-\vert\nabla u\vert^2g\right)
\end{equation*}
Plugging in $\nabla u$, this shows that $\nabla\vert\nabla u\vert^{2}$ and $\nabla u$ are parallel, more precisely that they satisfy \eqref{eq:uprop} which also appears in the proof of \Cref{thm:Robinsondeg}. We can then conclude as in said proof.
\end{proof}

\bibliographystyle{amsplain}
\bibliography{sample}

\end{document}